\pdfoutput=1

\documentclass[journal,10pt]{IEEEtran}
\usepackage{relsize}
\usepackage{amsmath}
\usepackage{multicol}
\usepackage{amsfonts}
\usepackage{amsmath}
\usepackage{graphicx}
\usepackage{textcomp}
\usepackage{float}
\usepackage{epstopdf}
\usepackage{mathtools}
\usepackage{multicol}
\usepackage{float}
\usepackage{lipsum}

\newtheorem{remark}{Remark}

\newtheorem{theorem}{Theorem}
\newtheorem{proof}{Proof}
  
\input{tcilatex}
\begin{document}

\title{Convergence Rate Analysis for Periodic Gossip Algorithms in Wireless
Sensor Networks}
\author{S. Kouachi, Sateeshkrishna Dhuli, and Y. N. Singh~%
\IEEEmembership{Senior Member,~IEEE} \thanks{%
S. Kouachi is with the Department of Mathematics, College of Science, Qassim
University, P. O. Box 6644, Al-Gassim, Buraydah 51452, Saudi Arabia e-mail:
koashy@qu.edu.sa}\thanks{%
Sateeshkrishna Dhuli is with the Department of Electrical Engineering,
Indian Institute of Technology,Kanpur, 2018016, Uttar pradesh, India e-mail:
dvskrishna.nitw@gmail.com}\thanks{%
Y.N.Singh is with the Department of Electrical Engineering, Indian Institute
of Technology,Kanpur, 2018016, Uttar pradesh, India e-mail:
ynsingh@iitk.ac.in}
}
\maketitle

\begin{abstract}
Periodic gossip algorithms have generated a lot of interest due to their ability to compute the global statistics by using local pairwise communications among nodes. Simple execution, robustness to topology changes, and distributed nature make these algorithms quite suitable for wireless sensor networks (WSN). However, these algorithms converge to the global statistics after certain rounds of pair-wise communications. A significant challenge for periodic gossip algorithms is difficult to predict the convergence rate for large-scale networks. To facilitate the convergence rate evaluation, we study a one-dimensional lattice network model.  In this scenario, to derive the explicit formula for convergence rate, we have to obtain a closed form expression for second largest eigenvalue of perturbed pentadiagonal matrices. In our approach, we derive the explicit expressions of eigenvalues by exploiting the theory of recurrent sequences. Unlike the existing methods in the literature, this is a direct method which avoids the theory of orthogonal polynomials \cite{hadj}. Finally, we derive the explicit expressions for convergence rate of the average periodic gossip algorithm in one-dimensional WSNs. We analyze the convergence rate by considering the linear weight updating approach and investigate the impact of gossip weights on the convergence rates for the different number of nodes. Further, we also study the effect of link failures on the convergence rate for average periodic gossip algorithms.
\end{abstract}

{} 


\IEEEpeerreviewmaketitle

\section{Introduction}

\IEEEPARstart{G}{ossip}  is a distributed operation which enables the sensor nodes to asymptotically to determine the average of their initial gossip variables. Gossip algorithms (\cite{boyd}, \cite{liu}, \cite{falsone}, \cite{mou}, \cite{anderson}, \cite{he}, \cite{zanaj}, \cite{dimakis}) have generated a lot of attention in the last decade due to their ability to achieve the global average using pairwise communications between nodes. In contrast to Centralized algorithms, the underlying distributed philosophy of these algorithms avoids the need for a fusion center for information gathering. Especially, they are quite suitable for data delivery in WSNs(\cite{mou}, \cite{anderson}, \cite{he}) as they can be utilized when the global network topology is highly dynamic, and network consists of power constrained nodes. As the gossip algorithms are iterative in nature, the convergence rate of the
algorithms greatly influences the performance of the WSNs. Although there have been several studies on gossip algorithms, analytic tools to control the convergence rate have not been much explored in the literature. \\
In a gossip algorithm, each node communicates information with one of the neighbors to obtain the global average at every node. Gossip algorithms have been shown to have faster convergence rates with the use of periodic gossip sequences, and such algorithms are termed as periodic gossip algorithms (\cite{zanaj}, \cite{dimakis}). Convergence rate of a periodic gossip algorithm is characterized by the magnitude of the second largest eigenvalue of a gossip matrix \cite{mou}. However, computing the second largest eigenvalue requires huge computational resources for large-scale networks. In our work, we estimate the convergence rate of the periodic gossip algorithms for one-dimensional Lattice network. Lattice networks represent the notion of geographical proximity in the practical WSNs, and they have been extensively used in the WSN applications for measuring and monitoring purposes \cite{wireless}.  Lattice networks (\cite{lattice}, \cite{el2}, \cite{spirakis}, \cite{li}, \cite{hekmat}, \cite{dousse}) are amenable to closed-form solutions which can be generalized to higher dimensions. These structures also play a fundamental role to analyze the connectivity, scalability, network size, and node failures in WSNs.  \\

In this paper, we model the WSN as a one-dimensional lattice network and obtain the explicit formulas of convergence rate for periodic gossip algorithms by considering both even and odd number of nodes. To obtain the convergence rate, we need to derive the explicit expressions of second largest eigenvalue for the perturbed pentadiagonal stochastic matrix. For properties of these matrices, we refer to (\cite{hadj}, \cite{salkuyeh}). In \cite{kilic}, the author considered a constant-diagonals matrix and gave many examples of the determinant and inverse of the matrix for some special cases. To the best of our knowledge, explicit expressions of eigenvalues for pentadiagonal matrices are not yet available in the literature. In our work, we derive the explicit expressions of eigenvalues for perturbed pentadiagonal matrices. Closed-form expressions of eigenvalues are extremely helpful as pentadiagonal have been widely used in the applications of time series analysis, signal processing, boundary value problems of partial differential equations, high order harmonic filtering theory, and differential equations (\cite{ravi}, \cite{capizzano}). To determine the eigenvalues, we obtain the recurrent relations followed by the application of the theory of recurrent sequences. Unlike the existing approaches in the literature, our approach avoids the theory of orthogonal polynomials. Furthermore, we use the explicit eigenvalue expressions of perturbed pentadiagonal matrices to derive the convergence rate expressions of periodic gossip algorithm in one-dimensional WSNs.  Specifically, our work avoids the usage of computationally expensive algorithms for studying the large-scale networks. Our results are more precise and they can be applied to most of the practical WSNs.
 
\subsection{Our Main Contributions}
(1)Firstly, we model the WSN as a one-dimensional lattice network and compute the gossip matrices of average periodic gossip algorithm for both even and odd number of nodes.
\newline
(2)To obtain the convergence rate of average periodic gossip algorithm, we need to determine the second largest eigenvalue of perturbed pentadiagonal matrices. By exploiting the theory of recurrent sequences, we derive the explicit expressions of eigenvalues for perturbed pentadiagonal matrices.
\newline
(3)We extend our results to periodic gossip algorithms with linear weight updating approach to obtain the generalized expression for convergence rate. 
\newline
(4)We consider the case of link failures and obtain the explicit expressions of convergence rate for average periodic gossip algorithms.
\newline
(5)Finally, we present the numerical results and study the effect of number of nodes and gossip weight on convergence rate. 

\subsection{Organization}

In summary, the paper is organized as follows. In Section II, we give the brief review of the periodic gossip algorithm. In Section III, we evaluate the primitive gossip matrices of periodic gossip algorithm for lattice networks. In Section IV, we derive the explicit eigenvalues of several perturbed pentadiagonal matrices using recurrent sequences. We derive the analytic expressions for convergence rate of periodic gossip algorithm for both even and odd number of nodes in Section V. Finally, in Section VI, we present the numerical results and study the effect of gossip weight and the number of nodes on the convergence rate.

\section{Brief Review of Periodic Gossip Algorithm}

Gossiping is a form of consensus to evaluate the global average of the
initial values of the gossip variables. The gossiping process can be modeled
as a discrete time linear system \cite{liu} as 
\begin{equation}
\textbf{x(t+1)}=M(t)\textbf{x(t)}, \quad t=1,2,..
\end{equation}
where $\textbf{x}$ is a vector of node variables, and $M(t)$ denotes a doubly stochastic
matrix. If nodes $i$ and $j$ gossip at time $t$, then the values of nodes at
time $(t+1)$ will be updated as 
\begin{equation}
x_{i}(t+1)=x_j(t+1)=\frac{x_i(t)+x_j(t)}{2}
\end{equation}
$M(t)$ is expressed as $M(t)$=$P_{i,j}$, where $P_{ij}$=$[P_{lm}]_{n%
\times n}$ for each step $(i,j)$ with entries defined as 
\begin{equation}
p_{lm}=\left\{%
\begin{matrix}
\frac{1}{2}, & (l,m)\in {(i,i),(i,j),(j,i),(j,j)} \\ 
1, & l=m,l\neq i,l \neq j ; \\ 
0, & otherwise.%
\end{matrix}%
\right.  \label{1}
\end{equation}
A gossip sequence is defined as an ordered sequence of edges for a given
graph in which each pair appears once. For a gossip sequence $%
(i_{1},j_{1}),(i_{2},j_{2}),.............(i_{k},j_{k})$, the gossip matrix
is expressed as $P_{i_{k}j_{k}}.......P_{i_{2}j_{2}}P_{i_{1}j_{1}}$. For a
periodic gossip sequence with period $T$, if $i_t,j_t$ denotes $t^{th}$ gossip
pair, then $i_{T+k}=i_{k}$ for $k=1,2,...$.\newline
Here, we can write variable $x$ at $(k+1)$ as 
\begin{equation}
\textbf{x((k+1)T)}=W\textbf{x(kT)},k=0,1,2...n,
\end{equation}
where \textit{W} is a doubly stochastic matrix. \textit{T} also denotes the number of steps needed to implement it's one period sub-sequence \textit{E}.\\
When a subset of edges are such that no two edges are adjacent to the same node and the gossips on these edges can be performed simultaneously in one time step is defined as multi-gossip. Minimum value of \textit{T} is related to an edge coloring problem. The minimum number of colors needed in an edge coloring problem is called as chromatic index. The value of the chromatic index is either $d_{max}$ or $d_{max}+1$, where $d_{max}$ is the maximum degree of a graph. When multi-gossip is allowed, a periodic gossip sequence $E,E,E$ with \textit{T}=chromatic index is called an optimal periodic gossip sequence. Convergence rate \cite{liu},\cite{mou} of the periodic gossip algorithm is characterized by the second largest eigenvalue $\left( {\lambda_2(W)} \right)$. Convergence rate (\textit{R}) at which gossip variable converges to a rank one matrix is determined by the spectral gap  $\left( {1-\lambda_2(W)} \right)$ \cite{boyd},\cite{dimakis}.

\begin{figure}[tbp]
\centering
\includegraphics[totalheight=1.2cm]{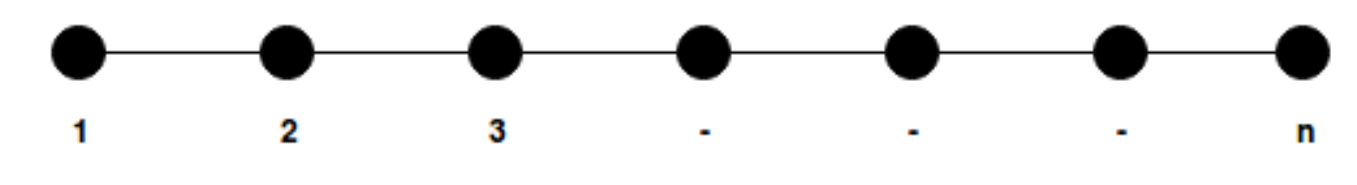}
\caption{One-Dimensional Lattice Network}
\label{fig:verticalcell}
\end{figure}

\section{Periodic Gossip Algorithm for a One-Dimensional Lattice Network}
We model the WSN as a one-dimensional Lattice network as shown in the Figure 
$1$. We obtain the optimal periodic sub-sequence and evaluate the primitive
gossip matrix.
\subsection{Average Gossip Algorithm}
In this algorithm, each pair of nodes at each
iteration participate in the gossip process to update with the average of their previous state values to obtain the global average. In this section, we study the average gossip algorithm for one-dimensional lattice network for both even and odd number of nodes.
\subsubsection{For $n$=even}
The possible pairs of one-dimensional lattice network can be expressed as
\begin{equation*}
\left \{
(1,2)(2,3)(3,4)...........(n-2,n-1),(n-1,n) \right \}
\end{equation*} 
In this case, the chromatic index is either $2$ or $3$. Hence, optimal periodic sub-sequence ($E$) can be written as 
\begin{equation*}
E=E_{1}E_{2},
\end{equation*}
where, $E_{1}=\left \{ (2,3) (4,5).....(n-2,n-1)\right \}$ 
and \\$%
E_{2}=\left
\{ (1,2)(3,4)(5,6)..........(n-1,n) \right \}$ are two disjoint sets. \newline
Primitive gossip matrix (\textit{W}) is expressed as \newline
$W=S_{1}S_{2}$ or $W=S_{2}S_{1}$, where \newline
$S_{1}=P_{2,3}P_{4,5}......P_{(n-2),(n-1)}$ \newline
$S_{2}=P_{1,2}P_{3,4}P_{5,6}....P_{(n-1),n}$ \newline
Hence, gossip matrix (\textit{W}) for $n$=even can be computed as 
\begin{equation*}
\resizebox{.8\hsize} {!} {$W=\left( \begin{array}{cccccc} \frac{1}{2} &
\frac{1}{4} & \frac{1}{4} & 0 & \cdots & 0 \\ \frac{1}{2} & \frac{1}{4} &
\ddots & 0 & \cdots & \vdots \\ 0 & \frac{1}{4} & \frac{1}{4} & \ddots &
\frac{1}{4} & \vdots \\ \vdots & \frac{1}{4} & \ddots & \frac{1}{4} &
\frac{1}{4} & 0 \\ \vdots & \cdots & 0 & \ddots & \frac{1}{4} & \frac{1}{2}
\\ 0 & \cdots & 0 & \frac{1}{4} & \frac{1}{4} &
\frac{1}{2}\end{array}\right) $}
\end{equation*}
\subsubsection{For $n$=odd}
In this case, optimal periodic sub-sequence (\textit{E}) is expressed as 
\begin{equation}
E=E_{1}E_{2},
\end{equation}
where, $E_{1}$=$\left \{ (2,3) (4,5).....(n-2,n-1)\right \}$ and \\
 $%
E_{2}$=$\left
\{ (1,2)(3,4)(5,6)..........(n-1,n) \right \}$ are two disjoint sets. \newline
primitive gossip matrix (\textit{W}) for $n$=odd is defined as \newline
$W=S_{1}S_{2}$ or $W=S_{2}S_{1}$, where \newline
$S_{1}=P_{1,2}P_{3,4}......P_{(n-2),(n-1)}$ \newline
$S_{2}=P_{2,3}P_{4,5}P_{6,7}....P_{(n-1),n}$ \newline \\
Hence, primitive gossip matrix (\textit{W}) for $n$=odd can be computed as 
\begin{equation*}
\resizebox{.8\hsize} {!} {$W=\left( \begin{array}{cccccc} \frac{1}{2} &
\frac{1}{4} & \frac{1}{4} & 0 & \cdots & 0 \\ \frac{1}{2} & \frac{1}{4} &
\ddots & 0 & \cdots & \vdots \\ 0 & \frac{1}{4} & \frac{1}{4} & \ddots &
\frac{1}{4} & \vdots \\ \vdots & \frac{1}{4} & \ddots & \frac{1}{4} &
\frac{1}{4} & 0 \\ \vdots & \cdots & 0 & \ddots & \frac{1}{4} & \frac{1}{4}
\\ 0 & \cdots & 0 & 0 & \frac{1}{2} & \frac{1}{2}\end{array}\right) $}
\end{equation*}
\subsection{Linear Weight Updating Approach}
In the previous section, we obtain the primitive gossip matrices for gossip weight 
\textit{w}=$\frac{1}{2}$. To investigate the effect of gossip weight on convergence
rate, we consider the special case by considering the weights associated
with the edges. If we assume that
at iteration \textit{k}, nodes \textit{i} and \textit{j} communicate, then node \textit{i} and node \textit{j} performs the
linear update with gossip weight \textit{w} as \cite{falsone}, \cite{mangoubi}.\newline
\begin{equation}
x_{i}(k)=(1-w)x_{i}(k-1)+w x_{j}(k-1)
\end{equation}
and 
\begin{equation}
x_{j}(k)=w x_{i}(k-1)+(1-w) x_{j}(k-1)
\end{equation}
, where \textit{w} is the gossip weight associated with edge $(i,j)$. We follow the similar steps as in the previous section to compute the primitive gossip matrix. \\ \\
For \textit{n}=even, primitive gossip matrix \textit{W} is expressed as

{\footnotesize 
\begin{equation}
\resizebox{1.02 \hsize} {!} {$W=\left( \begin{array}{cccccc} 1-w & -w^{2}+w
& w^{2} & 0 & \cdots & 0 \\ w & \left( w-1\right) ^{2} & \ddots & 0 & \cdots
& \vdots \\ 0 & -w^{2}+w & \left( w-1\right) ^{2} & \ddots & w^{2} & \vdots
\\ \vdots & w^{2} & \ddots & \left( w-1\right) ^{2} & -w^{2}+w & 0 \\ \vdots
& \cdots & 0 & \ddots & \left( w-1\right) ^{2} & w \\ 0 & \cdots & 0 & w^{2}
& -w^{2}+w & 1-w\end{array}\right) $}.  \label{2}
\end{equation}
}

Similarly, for $n$=odd, primitive gossip matrix (\textit{W}) can be computed as

{\footnotesize 
\begin{equation}
\resizebox{1.02 \hsize} {!} {$W=\left( \begin{array}{cccccc} 1-w & -w(w-1) &
w^{2} & 0 & \cdots & 0 \\ w & \left( w-1\right) ^{2} & \ddots & 0 & \cdots &
\vdots \\ 0 & -w(w-1)& \left( w-1\right) ^{2} & \ddots & w^{2} & \vdots \\
\vdots & w^{2} & \ddots & \left( w-1\right) ^{2} & -w^{2}+w & 0 \\ \vdots &
\cdots & 0 & \ddots & \left( w-1\right) ^{2} & -w(w-1)\\ 0 & \cdots & 0 & 0
& w & (1-w) \end{array} \right) $}  \label{3}
\end{equation}
}
\subsection{Effect of Link Failures on Convergence Rate}
Wireless sensor networks are prone to link failures due to noise, interference, and environmental changes. In this section, we study the effect of link failures on convergence rate for average periodic gossip algorithms. Let us consider the one-dimensional lattice network, where each link fails with the probability \textit{p}. \\
Primitive gossip matrix for even number of nodes is expressed as
\begin{equation}
\resizebox{1.03 \hsize} {!} {$P_{even}
=\left( \begin{array}{cccccccc} \frac{p+1}{2} & \frac{1-p^2}{4} & \frac{(1-p)^2}{4} & 0 & \ldots & \ldots &
\ldots & 0 \\ \frac{1-p}{2} & \frac{(p+1)^2}{4} & \frac{1-p^2}{2} & 0 & \ldots & \ldots & \ldots & \vdots \\ 0 & \frac{1-p^2}{4} & \frac{(p+1)^2}{4}
& \frac{1-p^2}{4} & \ddots & 0 & \ldots & \vdots \\ \vdots & \frac{(1-p)^2}{4} & \frac{1-p^2}{4} & \frac{(p+1)^2}{4} & \frac{1-p^2}{4} & \ddots &
\ldots & \vdots \\ \vdots & \ldots & \ddots & \frac{1-p^2}{4} & \frac{(p+1)^2}{4} & \frac{1-p^2}{4} & \frac{(1-p)^2}{4} & \vdots \\
\vdots & \ldots & 0 & \ddots & \frac{1-p^2}{4} & \frac{(p+1)^2}{4} & \frac{1-p^2}{4} & 0 \\ \vdots & \ldots & \ldots &
\ldots & 0 & \frac{1-p^2}{4} & \frac{(p+1)^2}{4} & \frac{1-p}{2} \\ 0 & \ldots & \ldots & \ldots & 0 & \frac{(p-1)^2}{4} & \frac{1-p^2}{4} &\frac{p+1}{2}
\end{array}\right) .$}  \label{4}
\end{equation}
Similarly, primitive gossip matrix for odd number of nodes is expressed as
\begin{equation}
\resizebox{1.03 \hsize} {!} {$P_{odd}
=\left( \begin{array}{ccccccccc} \frac{1+p}{2} & \frac{1-p^2}{4} & \frac{(p-1)^2}{4} & 0 & \ldots & \ldots &
\ldots & 0 & 0 \\ \frac{1-p}{2} & \frac{(p+1)^2}{4} & \frac{1-p^2}{4} & 0 & \ldots & \ldots & \ldots & \vdots & \vdots
\\ 0 & \frac{1-p^2}{4} & \frac{(p+1)^2}{4} & \frac{1-p^2}{4} & \ddots & 0 & \ldots & \vdots & \vdots \\ \vdots & \frac{(p-1)^2}{4} & \frac{1-p^2}{4} &
\frac{(p+1)^2}{4} & \frac{1-p^2}{4} & \ddots & \ldots & \vdots & \vdots \\ \vdots & \ldots & \ddots & \frac{1-p^2}{4} &
\frac{(p+1)^2}{4} & \frac{1-p^2}{4} & \frac{(p-1)^2}{4} & \vdots & \vdots \\ \vdots & \ldots & 0 & \ddots & \frac{1-p^2}{4}& \frac{(p+1)^2}{4} & \frac{(1-p^2)}{4} & 0
& 0 \\ \vdots & \ldots & \ldots & \ldots & 0 & \frac{1-p^2}{4} & \frac{(p+1)^2}{4} & \frac{1-p^2}{4} & \frac{(p-1)^2}{4} \\ 0 & \ldots &
\ldots & \ldots & 0 & \frac{(p-1)^2}{4} & \frac{1-p^2}{4} & \frac{(p+1)^2}{4} & \frac{1-p^2}{4} \\ 0 & \ldots & \ldots & \ldots & \ldots
& \ldots & 0 & \frac{1-p}{2} & \frac{1+p}{2}\end{array}\right) . $}  \label{5}
\end{equation}%
\section{Explicit Eigenvalues of perturbed pentadiagonal Matrices using recurrent sequences}
In this section, we derive the eigenvalues of the following
non-symmetric perturbed pentadiagonal matrices%
 \begin{equation}
\resizebox{1.03 \hsize} {!} {$A_{2m+1}\left( \alpha ,\beta ,e,bd,c,bd\right)
=\left( \begin{array}{ccccccccc} e-\alpha & b & c & 0 & \ldots & \ldots &
\ldots & 0 & 0 \\ d & e & b & 0 & \ldots & \ldots & \ldots & \vdots & \vdots
\\ 0 & b & e & b & \ddots & 0 & \ldots & \vdots & \vdots \\ \vdots & c & b &
e & b & \ddots & \ldots & \vdots & \vdots \\ \vdots & \ldots & \ddots & b &
e & b & c & \vdots & \vdots \\ \vdots & \ldots & 0 & \ddots & b & e & b & 0
& 0 \\ \vdots & \ldots & \ldots & \ldots & 0 & b & e & b & c \\ 0 & \ldots &
\ldots & \ldots & 0 & c & b & e & b \\ 0 & \ldots & \ldots & \ldots & \ldots
& \ldots & 0 & d & e-\beta\end{array}\right) . $}  \label{1.1}
\end{equation}%
and%
\begin{equation}
\resizebox{1.03 \hsize} {!} {$A_{2m}\left( \alpha ,\beta ,e,bd,c,bd\right)
=\left( \begin{array}{cccccccc} e-\alpha & b & c & 0 & \ldots & \ldots &
\ldots & 0 \\ d & e & b & 0 & \ldots & \ldots & \ldots & \vdots \\ 0 & b & e
& b & \ddots & 0 & \ldots & \vdots \\ \vdots & c & b & e & b & \ddots &
\ldots & \vdots \\ \vdots & \ldots & \ddots & b & e & b & c & \vdots \\
\vdots & \ldots & 0 & \ddots & b & e & b & 0 \\ \vdots & \ldots & \ldots &
\ldots & 0 & b & e & d \\ 0 & \ldots & \ldots & \ldots & 0 & c & b &
e-\beta\end{array}\right) .$}  \label{1}
\end{equation}

We apply the well known Gaussian elimination method to obtain the
characteristic polynomials of the matrices (\ref{1.1}, \ref{1}) as
orthogonal polynomials of second kind solutions of recurrent sequence
relations.

\begin{remark}
We observe that the presence of `e' on the main diagonal is redundant. In
fact, it can be considered as $0$. 
\end{remark}

\subsection{Case$\protect\alpha $=$\protect\beta =0$}

\subsubsection{Step 1}

In this step, we study the case when $d=b$. If we denote by $\Delta _{n}$
the characteristic polynomial of the matrix $A_{n}$,\ then the Laplace
expansion by minors along the last column provides

\begin{equation}
\Delta _{2m+2}=Y\Delta _{2m+1}-b\Delta _{2m+1}^{+},  \label{Rel bb 2}
\end{equation}%
where $Y=e-\lambda$.

Here, $\Delta _{2m+1}$ is the determinant of the matrix obtained from $%
A_{2m+2}-\lambda I_{2m+2}$ by deleting its last row and column and $\Delta
_{2m+1}^{+}$ is the determinant of the matrix obtained from $%
A_{2m+1}-\lambda I_{2m+1}$ by replacing the elements $c$ and $b$ with $b$
and $Y$ respectively in the last row. \newline
The Laplace expansion by minors of $\Delta _{2m+1}$ and $\Delta _{2m+1}^{+}$%
\ each one along the last row provides%

\begin{equation}
\Delta _{2m+1}=Y\Delta _{2m}-b\Delta _{2m}^{-}  \label{Rel bb 1}
\end{equation}%
and%
\begin{equation}
\Delta _{2m+1}^{+}=b\Delta _{2m}-c\Delta _{2m}^{-}  \label{Rel bb 1+}
\end{equation}%
where $\Delta _{2m}$ is the determinant of the matrix obtained from $%
A_{2m+1}-\lambda I_{2m+1}$ by deleting its last row and last column. $\Delta
_{2m}^{-}$ is the determinant of the matrix obtained from $A_{2m}-\lambda
I_{2m}$ by replacing its last column $b$ and $Y$ \ by $c$ and $b$
respectively.\newline
Concerning the determinant $\Delta _{2m}^{-},$\ the Laplace expansion by
minors along the last column provides%

\begin{equation}
\Delta _{2m}^{-}=b\Delta _{2m-1}-c\Delta _{2m-1}^{+}.  \label{Rel bb 1-}
\end{equation}%

Combining the above formulas, we obtain%

\begin{equation*}
\Delta _{n+2}-\left( Y^{2}+c^{2}-2b^{2}\right) \Delta _{n}+\left(
cY-b^{2}\right) ^{2}\Delta _{n-2}=0,n=3,...
\end{equation*}%

For convenience, we assume that $\Delta _{m}^{\left( 1\right) }=\Delta
_{2m+1}$ and $\Delta _{m}^{\left( 2\right) }=\Delta _{2m}$. Consequently, we have the following two recurrent sequences%

\begin{equation*}
\Delta _{m+1}^{\left( j\right) }-\left( Y^{2}+c^{2}-2b^{2}\right) \Delta
_{m}^{\left( j\right) }+\left( cY-b^{2}\right) ^{2}\Delta _{m-1}^{\left(
j\right) }=0,
\end{equation*}%
$j=1,2,..., m=2,...,$.

Here, characteristic polynomial is given by
\begin{equation*}
\alpha ^{2}-\left( Y^{2}+c^{2}-2b^{2}\right) \alpha +\left( cY-b^{2}\right)
^{2}=0.
\end{equation*}%
Substituting%

\begin{equation*}
\left( Y^{2}+c^{2}-2b^{2}\right) ^{2}-4\left( cY-b^{2}\right) ^{2}=-4\left(
cY-b^{2}\right) ^{2}\sin ^{2}\theta ,
\end{equation*}%

which is equivalent to the following algebraic equation%
\begin{equation}
Y^{2}-2cY\cos \theta +c^{2}-2b^{2}\left( 1-\cos \theta \right) =0,  \label{&}
\end{equation}%
with roots%
\begin{equation*}
\alpha =\left( cY-b^{2}\right) e^{\pm i\theta }.
\end{equation*}%
Here, we denote $z$ as
\begin{equation*}
z=\left( cY-b^{2}\right)
\end{equation*}%
Finally, we can write the general solution of the above recurrent sequences as

\begin{equation*}
\Delta _{m}^{\left( j\right) }=z^{m}\left( C_{1}e^{-im\theta
}+C_{2}e^{im\theta }\right) ,\ \ j=1,2\text{ and}\ \ m=1,...,
\end{equation*}%

with the initial conditions%

\begin{equation*}
\Delta _{1}^{\left( j\right) }=Y,\ \ \ \Delta _{2}^{\left( j\right)
}=Y^{3}-2Yb^{2}+cb^{2},
\end{equation*}%
when $j=1$ and%

\begin{equation*}
\Delta _{1}^{\left( j\right) }=Y^{2}-b^{2},\ \ \ \Delta _{2}^{\left(
j\right) }=Y^{4}-3Y^{2}b^{2}+2Yb^{2}c+b^{4}-b^{2}c^{2},
\end{equation*}%

when $j=2$ and where $C_{1}$ and $C_{2}$ are real constants to be calculated.
Our main result in this section is as follows \\

\begin{theorem}
The characteristic polynomial $\Delta _{m}^{\left( j\right) }$ is%
\begin{equation}
\Delta _{m}^{\left( 1\right) }=z^{m}\frac{Y\sin \left( m+1\right) \theta
-c\sin m\theta }{\sin \theta },  \label{For bb 1}
\end{equation}%
and%
\begin{equation}
\Delta _{m}^{\left( 2\right) }=z^{m-1}\frac{z\sin \left( m+1\right) \theta
+\left( b^{2}-c^{2}\right) \sin m\theta }{\sin \theta }.  \label{For bb 2}
\end{equation}
\end{theorem}

\begin{proof}
When $j=1$, using initial conditions and calculating the constants $C_{1}$
and $C_{2}$, we obtain 
\begin{eqnarray*}
\frac{\Delta _{m}^{\left( 1\right) }}{z^{m-1}} &=&\frac{-zY\sin \left(
m-1\right) \theta +\left( Y^{3}-2Yb^{2}+cb^{2}\right) \sin m\theta }{\sin
\theta } \\
&=&\frac{E}{\sin \theta },
\end{eqnarray*}%
where%

\begin{equation*}
E=Y\left[ \left( Y^{2}-b^{2}\right) \sin m\theta -z\sin \left( m-1\right)
\theta \right] -b^{2}\left( Y-c\right) \sin m\theta .
\end{equation*}
From (\ref{&}), we have 
\begin{equation*}
Y^{2}-b^{2}=2z\cos \theta +b^{2}-c^{2},
\end{equation*}%
We use the following trigonometric identity%
\begin{equation}
2\cos \theta \sin m\theta -\sin \left( m-1\right) \theta =\sin \left(
m+1\right) \theta ,  \label{Trig}
\end{equation}%

Hence from (\ref{Trig}), we get

\begin{eqnarray*}
E &=&Yz\sin \left( m+1\right) \theta +\left[ \left( b^{2}-c^{2}\right)
Y-b^{2}\left( Y-c\right) \right] \sin m\theta \\
&=&z\left[ Y\sin \left( m+1\right) \theta -c\sin m\theta \right] .
\end{eqnarray*}%

Finally, we get (\ref{For bb 1}).\newline

For the $j=2$ case, we follow the similar steps as in the first case and obtain  

\begin{eqnarray*}
\frac{\Delta _{m}^{\left( 2\right) }}{z^{m-2}} &=&\frac{-z\Delta _{2}\sin
\left( m-2\right) \theta +\Delta _{4}\sin \left( m-1\right) \theta }{\sin
\theta } \\
&=&\frac{F}{\sin \theta },
\end{eqnarray*}%
where%
\begin{eqnarray*}
F &=&\left( Y^{2}-b^{2}\right) \left[ \left( Y^{2}-b^{2}\right) \sin \left(
m-1\right) \theta -z\sin \left( m-2\right) \theta \right] \\
&&-b^{2}\left( Y-c\right) ^{2}\sin \left( m-1\right) \theta .
\end{eqnarray*}%
Using (\ref{&}) and (\ref{Trig}), we obtain
\begin{eqnarray*}
F &=&\left( Y^{2}-b^{2}\right) z\sin m\theta \\
&&+\left[ \left( b^{2}-c^{2}\right) \left( Y^{2}-b^{2}\right) -b^{2}\left(
Y-c\right) ^{2}\right] \sin \left( m-1\right) \theta \\
&=&\left( Y^{2}-b^{2}\right) z\sin m\theta -z^{2}\sin \left( m-1\right)
\theta \\
&=&z\left[ \left( Y^{2}-b^{2}\right) \sin m\theta -z\sin \left( m-1\right)
\theta \right] .
\end{eqnarray*}%

Using (\ref{&}), we have%

\begin{equation*}
F=z\left[ \left( 2z\cos \theta +b^{2}-c^{2}\right) \sin m\theta -z\sin
\left( m-1\right) \theta \right]
\end{equation*}%

The trigonometric identity (\ref{Trig}), gives%

\begin{eqnarray*}
F &=&z^{2}\sin \left( m+1\right) \theta +\left( b^{2}-c^{2}\right) z\sin
m\theta \\
&=&z\left[ z\sin \left( m+1\right) \theta +\left( b^{2}-c^{2}\right) \sin
m\theta \right] .
\end{eqnarray*}%

Finally, we get (\ref{For bb 2}).
\end{proof}

\subsubsection{Step 2}

In this step, we derive the characteristic polynomial of the matrix
obtained from $A_{n}$ by taking in the upper corner $d=b$ which is
denoted as $\Delta _{n}^{\left( j\right) }\left( bb,bd\right) ,\ j=1,\ 2$.
Here, the calculations differ only from the first step. We start 
developing the characteristic polynomial of $A_{n}$\ along the last column.
The relationships that follow are the same as the previous step. Finally, we
get 

\begin{equation}
b\left( Y-c\right) \Delta _{m+1}^{\left( j\right) }\left( bb,bd\right)
=R\Delta _{m}^{\left( j\right) }\left( bb,bb\right) -S\Delta
_{m-1}^{\left( j\right) }\left( bb,bb\right) ,  \label{For gen}
\end{equation}%
where%
\begin{equation}
\left\{ 
\begin{array}{l}
R=z\left( bY-cd\right) +\left( Y-c\right) ^{2}\left( bY-cd\right) +d\left(
Y-c\right) z, \\ 
\\ 
S=z^{2}\left( bY-cd\right) .%
\end{array}%
\right.  \label{R-S}
\end{equation}%
Using (\ref{&}), we have%
\begin{equation}
\left\{ 
\begin{array}{l}
R=z\left[ 2\left( bY-cd\right) \cos \theta +\left( d\mathbf{-}b\right) Y%
\right] , \\ 
S=z^{2}\left( bY-cd\right) .%
\end{array}%
\right.  \label{R}
\end{equation}%
The trigonometric identity (\ref{Trig}), gives 
\begin{equation}
\Delta _{m+1}^{\left( j\right) }\left( bb,bd\right) =K\Delta _{m+1}^{\left(
j\right) }+Lz\Delta _{m}^{\left( j\right) },  \label{New For 1}
\end{equation}%
where%
\begin{equation}
K=\frac{bY-cd}{b\left( Y-c\right) },\ \ \ L=\frac{\left( d\mathbf{-}b\right)
Y}{b\left( Y-c\right) }.  \label{KL}
\end{equation}%
\newline
\begin{theorem}
\bigskip The characteristic polynomial of the matrices $A_{2m+1}\left(
bb,bd\right) $ and $A_{2m}\left( bb,bd\right) $ are%
\begin{equation}
\resizebox{1.04 \hsize} {!} {$\Delta _{m+1}^{\left( 1\right) }\left(
bb,bd\right) =\frac{z^{m}}{\sin \theta }\left\{ Yz\sin \left( m+2\right)
\theta -\left[ cz+\left( d\mathbf{-}b\right) b\left( Y-c\right) \right] \sin
\left( m+1\right) \theta \right\}.$}  \label{New Exp 1}
\end{equation}%
and%
\begin{equation}
\resizebox{1.04 \hsize} {!} {$\Delta _{m+1}^{\left( 2\right) }\left(
bb,bd\right) =\frac{z^{m}}{\sin \theta }\left\{ \begin{array}{l} z\sin
\left( m+2\right) \theta -\left[ \left( d\mathbf{-}2b\right) b+c^{2}\right]
\sin \left( m+1\right) \theta \\ +\left( d\mathbf{-}b\right) b\sin m\theta
\end{array}\right\} ,$}  \label{Exp 2}
\end{equation}%
respectively.
\end{theorem}

\begin{proof}
Substituting the $(\ref{For bb 2})$ and $(\ref{New For 1})$ in $(\ref{For bb 1})$ results in 
\begin{equation*}
\resizebox{1.04 \hsize} {!} {$\Delta _{m+1}^{\left( 1\right) }\left(
bb,bd\right) =z^{m+1}\frac{B_{m+2}^{\left( 1\right) }\sin \left( m+2\right)
\theta +B_{m+1}^{\left( 1\right) }\sin (m+1)\theta +B_{m}^{\left( 1\right)
}\sin m\theta }{b\left( Y-c\right) \sin \theta },$}
\end{equation*}%
and
\begin{equation*}
\resizebox{1.04 \hsize} {!} {$\Delta _{m+1}^{\left( 2\right) }\left(
bb,bd\right) =z^{m}\frac{B_{m+2}^{\left( 2\right) }\sin \left( m+2\right)
\theta +B_{m+1}^{\left( 2\right) }\sin (m+1)\theta +B_{m}^{\left( 2\right)
}\sin m\theta }{b\left( Y-c\right) \sin \theta },$}
\end{equation*}%
where%
\begin{eqnarray*}
B_{m+2}^{\left( 1\right) } &=&\left( bY-cd\right) Y,\ \ \ B_{m}^{\left(
1\right) }=-\left( d\mathbf{-}b\right) cY \\
B_{m+1}^{\left( 1\right) } &=&-\left( bY-cd\right) c+\left( d\mathbf{-}%
b\right) Y^{2},
\end{eqnarray*}%
and%
\begin{eqnarray*}
B_{m+2}^{\left( 2\right) } &=&\left( bY-cd\right) z,\ \ B_{m}^{\left(
2\right) }=\left( d\mathbf{-}b\right) Y\left( b^{2}-c^{2}\right) \\
B_{m+1}^{\left( 2\right) } &=&\left( bY-cd\right) \left( b^{2}-c^{2}\right)
+\left( d\mathbf{-}b\right) Yz.
\end{eqnarray*}%
Using 
\begin{equation*}
\sin m\theta +\sin \left( m+2\right) \theta =2\cos \theta \sin \left(
m+1\right) \theta ,
\end{equation*}%
the coefficients of $\Delta _{m+1}^{\left( 1\right) }\left( bb,bd\right) $
become%
\begin{eqnarray*}
B_{m+2}^{\left( 1\right) } &=&b\left( Y-c\right) Y,\ \ \ \ \ \
B_{m}^{\left( 1\right) }=0 \\
B_{m+1}^{\left( 1\right) } &=&-\left( bY-cd\right) c+\left( d\mathbf{-}%
b\right) Y^{2}-2\left( d\mathbf{-}b\right) cY\cos \theta
\end{eqnarray*}%
Applying (\ref{&}) to the coefficient $B_{m+1}^{\left( 1\right) }$ becomes%
\begin{equation*}
B_{m+1}^{\left( 1\right) }=b\left( Y-c\right) c-2\left( d\mathbf{-}b\right)
\left( \cos \theta -1\right) b^{2}.
\end{equation*}%
Using (\ref{&}), we get%
\begin{equation*}
B_{m+1}^{\left( 1\right) }=b\left( Y-c\right) c-\left( d\mathbf{-}b\right)
b^{2}\frac{\left( Y-c\right) ^{2}}{z}.
\end{equation*}%
Substituting in the expression of $\Delta _{m+1}^{\left( 1\right) }\left(
bb,bd\right) $ and simplifying the term $b\left( Y-c\right) $, we get (\ref%
{New Exp 1}).\newline
We use (\ref{&}) for $B_{m+1}^{\left( 2\right) }$, results in 
\begin{eqnarray*}
B_{m+1}^{\left( 2\right) } &=&\left( bY-cd\right) \left( b^{2}-c^{2}\right)
+\left( d\mathbf{-}b\right) cY^{2}-\left( d\mathbf{-}b\right) Yb^{2} \\
&=&-b\left( Y-c\right) \left[ \left( d\mathbf{-}2b\right) b+c^{2}\right]
+2\left( d\mathbf{-}b\right) cz\cos \theta .
\end{eqnarray*}%
Using (\ref{Trig}), we get%
\begin{eqnarray*}
B_{m+2}^{\left( 2\right) } &=&b\left( Y-c\right) z,\ \ \ \ B_{m}^{\left(
2\right) }=\left( d\mathbf{-}b\right) b^{2}\left( Y-c\right) , \\
B_{m+1}^{\left( 2\right) } &=&-b\left( Y-c\right) \left[ \left( d\mathbf{-}%
2b\right) b+c^{2}\right] ,\ \ 
\end{eqnarray*}%
Finally, we can simplify as (\ref{Exp 2}).
\end{proof}
\begin{remark}
In the subsequent sections, we use the following expression instead of (\ref%
{New Exp 1} )%
\begin{equation}
\resizebox{1.04 \hsize} {!} {$\Delta _{m+1}^{\left( 1\right) }\left(
bb,bd\right) =\frac{z^{m-1}}{b\sin \theta }\left. \begin{array}{l} z\left\{
\begin{array}{l} \left[ b\left( Y+c\right) -cd\right] \sin \left( m+2\right)
\theta \\ +\left[ \left( d\mathbf{-}b\right) \left( Y+c\right) -bc\right]
\sin \left( m+1\right) \theta -c\left( d\mathbf{-}b\right) \sin
m\theta\end{array}\right\} \\ -\left( d\mathbf{-}b\right)
\mathbf{c}^{2}\left( Y-c\right) \sin \left( m+1\right) \theta
,\end{array}\right. $}  \label{Exp 1}
\end{equation}%
which gives (\ref{New Exp 1}) by eliminating the term $-c\left( d\mathbf{-}%
b\right) \sin m\theta $ using trigonometric formulas.
\end{remark}

\subsubsection{Step 3}

In this step, we are interested in the calculus of the characteristic
polynomial of the matrices when $b\neq d$ in the two corners, which will be
denoted by $\Delta _{n}^{\left( j\right) }\left( bd,bd\right) ,\ j=1,\ 2$.
Applying the same reasoning as previous steps, we obtain 

\begin{equation}
\Delta _{m+1}^{\left( j\right) }\left( bd,bd\right) =K.\Delta _{m+1}^{\left(
j\right) }\left( bb,bd\right) +L.z\Delta _{m}^{\left( j\right) }\left(
bb,bd\right) ,\ \ j=1,2,  \label{New For 2}
\end{equation}%
where $K$ and $L$ are given by (\ref{KL}).

\begin{figure*}[!t]
{\normalsize 
}
\par
{\normalsize 
 }
\par
{\normalsize 
}
\par
\begin{theorem}
{\normalsize \bigskip The characteristic polynomial of the matrices $%
A_{2m+1}\left( bd,bd\right) $ are%
\begin{equation}
\Delta_{m + 1}^{(1)} (bd,bd) = z^m \frac{{\left\{ {Y(z + c^2 )\sin (m + 2)\theta  - c\left[ {z + 2b(Y - c) - c^2  + 2cY\cos \theta } \right]\sin (m + 1)\theta } \right\}}}{{sin\theta }} \label{Exp 1 bd bd}
\end{equation}
and%
\begin{eqnarray}
\sin \theta \frac{\Delta _{m+1}^{\left( 2\right) }\left( bd,bd\right) }{%
z^{m-1}} &=&z\left\{ 
\begin{array}{c}
z\sin \left( m+2\right) \theta +\left( 3b^{2}-2bd-c^{2}\right) \sin \left(
m+1\right) \theta \\ 
-\left( d\mathbf{-}b\right) \left( d\mathbf{-}3b\right) \sin m\theta +\left(
d\mathbf{-}b\right) ^{2}\sin \left( m-1\right) \theta%
\end{array}%
\right\}  \label{Exp 2 bddb} \\
&&+\left( d\mathbf{-}b\right) ^{2}c\left( Y-c\right) \sin m\theta .  \notag
\end{eqnarray}
}
\end{theorem}
\par
{\normalsize 
\hrulefill 
\vspace*{4pt} }
\end{figure*}

\begin{proof}
When $n$ is odd, we substitute the $\Delta _{m}^{\left( 1\right) }\left( bb,bd\right) $ and $\Delta
_{m-1}^{\left( 1\right) }\left( bb,bd\right) $\  in $(\ref{New For 2})$ to obtain 
\begin{equation*}
\resizebox{1.02 \hsize} {!} {$\Delta _{m+1}^{\left( 1\right) }\left(
bd,bd\right) =z^{m+1}\frac{D_{m+2}^{\left( 1\right) }\sin \left( m+2\right)
\theta +D_{m+1}^{\left( 1\right) }\sin (m+1)\theta +D_{m}^{\left( 1\right)
}\sin m\theta }{b\left( Y-c\right) \sin \theta },$}
\end{equation*}%
where%
\begin{equation*}
\resizebox{1.02 \hsize} {!} {$\left\{ \begin{array}{l} D_{m+2}^{\left(
1\right) }=\left( bY-cd\right) Yz, \\ D_{m+1}^{\left( 1\right) }=-\left[
c\left( bY-cd\right) z+\left( d\mathbf{-}b\right) b\left( bY-cd\right)
\left( Y-c\right) -\left( d\mathbf{-}b\right) Y^{2}z\right] , \\
D_{m}^{\left( 1\right) }=-\left( d\mathbf{-}b\right) Y\left[ cz+\left(
d\mathbf{-}b\right) b\left( Y-c\right) \right] . \end{array}\right. $}
\end{equation*}%
Assume $d-b=c$, then%
\begin{equation*}
D_{m+2}^{\left( 1\right) }-D_{m}^{\left( 1\right) }=b\left( Y-c\right) Y 
\left[ z+c^{2}\right] =b\left( Y-c\right) Y\left[ z+c^{2}\right] .
\end{equation*}%
Using trigonometric formulas, we obtain 
\begin{equation*}
D_{m}^{\left( 1\right) }\left[ \sin \left( m+2\right) \theta +\sin m\theta %
\right] =2D_{m}^{\left( 1\right) }\sin \left( m+1\right) \theta \cos \theta ,
\end{equation*}%
then%
\begin{equation*}
\resizebox{1.02 \hsize} {!} {$\Delta _{m+1}^{\left( 1\right) }\left(
bd,bd\right) =z^{m+1}\frac{\left\{ \begin{array}{l} \left[ D_{m+2}^{\left(
1\right) }-D_{m}^{\left( 1\right) }\right] \sin \left( m+2\right) \theta \\
+\left[ D_{m+1}^{\left( 1\right) }+2D_{m}^{\left( 1\right) }\cos \theta
\right] \sin (m+1)\theta \end{array}\right\} }{b\left( Y-c\right) \sin
\theta },$}
\end{equation*}%
We have%
\begin{equation*}
\resizebox{1.02 \hsize} {!} {$D_{m+1}^{\left( 1\right) }+2D_{m}^{\left(
1\right) }\cos \theta =\frac{\left\{ \begin{array}{l} \left[ b\left(
Y-c\right) -c^{2}-Y^{2}+2cY\cos \theta \right] cz \\ +bc\left( Y-c\right)
\left[ b\left( Y-c\right) -c^{2}+2cY\cos \theta \right] \end{array}\right\}
}{b\left( Y-c\right) }.$}
\end{equation*}%
Using (\ref{&}) 
\begin{equation*}
\resizebox{1.02 \hsize} {!} {$D_{m+1}^{\left( 1\right) }+2D_{m}^{\left(
1\right) }\cos \theta =\frac{\left\{ \begin{array}{l} \left[ b\left(
Y-c\right) +2b^{2}\left( \cos \theta -1\right) \right] cz \\ +bc\left(
Y-c\right) \left[ b\left( Y-c\right) -c^{2}+2cY\cos \theta \right]
\end{array}\right\} }{b\left( Y-c\right) .}.$}
\end{equation*}%
Using again (\ref{&}) 
\begin{equation*}
D_{m+1}^{\left( 1\right) }+2D_{m}^{\left( 1\right) }\cos \theta =c\left\{
z+2b\left( Y-c\right) -c^{2}+2cY\cos \theta \right\} .
\end{equation*}%
This gives (\ref{Exp 1 bd bd}).\newline
For $n$ even, we substitute in (\ref{New For 2})$\ $the quantities $\Delta
_{m}^{\left( 2\right) }\left( bb,bd\right) $ and $\Delta _{m-1}^{\left(
2\right) }\left( bb,bd\right) $\ by theirs expressions given by (\ref{Exp 2}%
), we get%
\begin{equation*}
\resizebox{1.02 \hsize} {!} {$\Delta _{m+1}^{\left( 2\right) }\left(
bd,bd\right) =z^{m}\frac{\left\{ \begin{array}{l} D_{m+2}^{\left( 2\right)
}\sin \left( m+2\right) \theta +D_{m+1}^{\left( 2\right) }\sin (m+1)\theta
\\ +D_{m}^{\left( 2\right) }\sin m\theta +D_{m-1}^{\left( 2\right) }\sin
\left( m-1\right) \theta \end{array}\right\} }{b\left( Y-c\right) \sin
\theta }.$}
\end{equation*}%
In order to simplify the above expression, we use%
\begin{equation*}
\left. 
\begin{array}{l}
\Psi ^{\left( 2\right) }\left( Y\right) =b\left( Y-c\right) \sin \theta 
\frac{\Delta _{m+1}^{\left( 2\right) }\left( bd,bd\right) }{z^{m+1}} \\ 
=D_{m+2}^{\left( 2\right) }\sin \left( m+2\right) \theta +D_{m+1}^{\left(
2\right) }\sin (m+1)\theta \\ 
+D_{m}^{\left( 2\right) }\sin m\theta +D_{m-1}^{\left( 2\right) }\sin \left(
m-1\right) \theta .%
\end{array}%
\right.
\end{equation*}%
We use (\ref{For gen}) to simplify%
\begin{equation*}
\resizebox{1.02 \hsize} {!} {$\Psi ^{\left( 2\right) }\left( Y\right)
=\left\{ \begin{array}{l} \frac{R}{\sin \theta }z^{m}\left\{
\begin{array}{l} z\sin \left( m+1\right) \theta -\left[ \left(
d\mathbf{-}2b\right) b+c^{2}\right] \sin m\theta \\ +\left(
d\mathbf{-}b\right) b\sin \left( m-1\right) \theta \end{array}\right\} \\
-\frac{S}{\sin \theta }z^{m-1}\left\{ \begin{array}{l} z\sin m\theta -\left[
\left( d\mathbf{-}2b\right) b+c^{2}\right] \sin \left( m-1\right) \theta \\
+\left( d\mathbf{-}b\right) b\sin \left( m-2\right) \theta
\end{array}\right\} .\end{array}\right. $}
\end{equation*}%
Using the expressions of $R$ and $S$ given by (\ref{R-S})\ and applying the
trigonometric identity (\ref{Trig}), we get%
\begin{equation*}
\resizebox{1.2 \hsize} {!} {$\begin{array}{lcl}
\Psi ^{\left( 2\right) }\left( Y\right) & = & \left( bY-cd\right) \left\{ z\sin
\left( m+2\right) \theta -\left[ \left( d\mathbf{-}2b\right) b+c^{2}\right]
\sin \left( m+1\right) \theta +\left( d\mathbf{-}b\right) b\sin m\theta
\right\}  \\
& & +\left( d\mathbf{-}b\right) Y\left\{ z\sin \left( m+1\right) \theta -\left[
\left( d\mathbf{-}2b\right) b+c^{2}\right] \sin m\theta +\left( d\mathbf{-}%
b\right) b\sin \left( m-1\right) \theta \right\} . 
\end{array}$} 
\end{equation*}
Applying (\ref{Trig}), we get%
\begin{equation*}
\resizebox{1.02 \hsize} {!} {$\Psi ^{\left( 2\right) }\left( c\right)
=2\left( \cos \theta -1\right) \left( d\mathbf{-}b\right) c\left\{ \left(
b^{2}-c^{2}\right) \sin \left( m+1\right) \theta +\left( d\mathbf{-}b\right)
b\sin m\theta \right\} .$}
\end{equation*}%
But from (\ref{&}), we have%
\begin{equation*}
2\left( \cos \theta -1\right) =\frac{\left( Y-c\right) ^{2}}{z},
\end{equation*}%
then%
\begin{equation*}
\left. 
\begin{array}{l}
\Psi ^{\left( 2\right) }\left( Y\right) =:\Psi ^{\left( 2\right) }\left(
Y\right) -\Psi ^{\left( 2\right) }\left( c\right) +\Psi ^{\left( 2\right)
}\left( c\right) \\ 
=\left( Y-c\right) \left\{ 
\begin{array}{l}
\left[ bc\left( Y+c\right) -b^{3}-c^{2}d\right] \sin \left( m+2\right) \theta
\\ 
-\left[ \left( d\mathbf{-}2b\right) b+c^{2}\right] b\sin \left( m+1\right)
\theta \\ 
+\left( d\mathbf{-}b\right) b^{2}\sin m\theta%
\end{array}%
\right\} \\ 
+\left( d\mathbf{-}b\right) \left( Y-c\right) \left\{ 
\begin{array}{l}
\left[ c\left( Y+c\right) -b^{2}\right] \sin \left( m+1\right) \theta \\ 
-\left[ \left( d\mathbf{-}2b\right) b+c^{2}\right] \sin m\theta \\ 
+\left( d\mathbf{-}b\right) b\sin \left( m-1\right) \theta%
\end{array}%
\right\} \\ 
+\left( d\mathbf{-}b\right) c\left\{ 
\begin{array}{c}
\left( b^{2}-c^{2}\right) \sin \left( m+1\right) \theta \\ 
+\left( d\mathbf{-}b\right) b\sin m\theta%
\end{array}%
\right\} \frac{\left( Y-c\right) ^{2}}{z}.%
\end{array}%
\right.
\end{equation*}%
That is%
\begin{equation*}
\left. 
\begin{array}{l}
b\sin \theta \tfrac{\Delta _{m+1}^{\left( 2\right) }\left( bd,bd\right) }{%
z^{m+1}} \\ 
=\left\{ 
\begin{array}{l}
\left[ bz-\left( d\mathbf{-}b\right) c^{2}\right] \sin \left( m+2\right)
\theta \\ 
-\left[ \left( d\mathbf{-}2b\right) b+c^{2}\right] b\sin \left( m+1\right)
\theta \\ 
+\left( d\mathbf{-}b\right) b^{2}\sin m\theta%
\end{array}%
\right\} \\ 
+\left( d\mathbf{-}b\right) \left\{ 
\begin{array}{l}
\left[ z+c^{2}\right] \sin \left( m+1\right) \theta \\ 
-\left[ \left( d\mathbf{-}2b\right) b+c^{2}\right] \sin m\theta \\ 
+\left( d\mathbf{-}b\right) b\sin \left( m-1\right) \theta%
\end{array}%
\right\} \\ 
+\left( d\mathbf{-}b\right) c\left\{ 
\begin{array}{c}
\left( b^{2}-c^{2}\right) \sin \left( m+2\right) \theta + \\ 
\left( d\mathbf{-}b\right) b\sin \left( m+1\right) \theta%
\end{array}%
\right\} \frac{\left( Y-c\right) }{z}.%
\end{array}%
\right.
\end{equation*}%
This ends the proof of the Theorem $3$.
\end{proof}
\subsection{General case}
The calculus when $\alpha $ and $\beta $ are nondescriptive become very
complicated and the expressions of the characteristic polynomials are very long.
In this section, we deal the case of $\alpha =\beta =-b$.
\begin{theorem}
The eigenvalues of the matrices are the couples $\lambda _{i,k}=e-Y_{i,k},\
i=1,2$ and $k=0,1,...,m-1$, where%
\begin{equation}
\left\{ 
\begin{array}{l}
Y_{1,0}=-\left( 2b+c\right) , \\ 
Y^2_{i,k}-2\left( cY_{i,k}-b^{2}\right) \cos \tfrac{\left( 2k+1\right) \pi }{%
2m+1}-\left( 2b^{2}-c^{2}\right) =0,\ \ \ \ 
\end{array}%
\right.   \label{Eig. 1}
\end{equation}%
$k=0,1,2,...,m-1,$ when $n$ is odd and%
\begin{equation}
\left\{ 
\begin{array}{l}
Y_{1,0}=c,\ \ Y_{2,0}=-\left( 2b+c\right) , \\ 
Y^2_{i,k}-2\left( cY_{i,k}-b^{2}\right) \cos \left( \tfrac{k\pi }{m}\right)
-\left( 2b^{2}-c^{2}\right) =0\ \ 
\end{array}%
\right.   \label{Eig. 2}
\end{equation}%
$,\ \ k=1,2,...,m-1,$ when $n$ is even.
\end{theorem}

\begin{proof}
When $n$ is odd, $\alpha \neq 0$ and $\beta \neq 0$, then expanding the
determinant of the matrix $A_{2m+1}\left( \alpha ,\beta ,Y,bd,c,bd\right)
-\lambda I_{2m+1}$ in terms of the first and last columns and using the
linear property of the determinants with regard to its columns results in%
\begin{equation*}
\left. 
\begin{array}{l}
\Delta _{m+1}^{\left( 1\right) }\left( \alpha ,\beta ,bd,bd\right) =\Delta
_{m+1}^{\left( 1\right) }\left( bd,bd\right) -\alpha \Delta _{m+1}^{\left(
2\right) }\left( bb,bd\right)  \\ 
-\beta \Delta _{m+1}^{\left( 2\right) }\left( bd,bb\right) +\alpha \beta
\Delta _{m}^{\left( 1\right) }\left( bb,bb\right) .%
\end{array}%
\right. 
\end{equation*}%
Taking $\alpha =\beta =-b\ $and substituting the expressions of $\Delta
_{m+1}^{\left( 1\right) }\left( bd,bd\right) $, $\Delta _{m+1}^{\left(
2\right) }\left( bb,bd\right) $ and $\Delta _{m-1}^{\left( 2\right) }\left(
bb,bb\right) $ given by (\ref{Exp 1 bd bd}), (\ref{Exp 2}) and (\ref{For bb
2}) respectively, we get%
\begin{equation*}
\resizebox{1.02 \hsize} {!} {$\Delta _{m+1}^{\left( 1\right) }\left( \alpha
,\beta ,bd,bd\right) =z^{m}\tfrac{E_{m+2}^{\left( 1\right) }\sin \left(
m+2\right) \theta +E_{m+1}^{\left( 1\right) }\sin (m+1)\theta +E_{m}^{\left(
1\right) }\sin m\theta }{\sin \theta },$}
\end{equation*}%
where%
\begin{equation*}
\left\{ 
\begin{array}{l}
E_{m+2}^{\left( 1\right) }=Y\left( z+c^{2}\right) +2bz, \\ 
E_{m+1}^{\left( 1\right) }=-\left[ 
\begin{array}{c}
cz+2bc\left( Y-c\right) -c^{3}+2c^{2}Y\cos \theta  \\ 
+2b\left( bc+c^{2}-b^{2}\right) -b^{2}Y%
\end{array}%
\right] , \\ 
E_{m}^{\left( 1\right) }=b^{2}c.%
\end{array}%
\right. 
\end{equation*}%
Using 
\begin{equation*}
b^{2}c\left[ \sin m\theta +\sin \left( m+2\right) \theta \right]
=2b^{2}c\sin \left( m+1\right) \theta \cos \theta ,
\end{equation*}%
then%
\begin{equation*}
\Delta _{m+1}^{\left( 1\right) }\left( \alpha ,\beta ,bd,bd\right) =P\sin
\left( m+2\right) \theta -Q\sin \left( m+1\right) \theta ,
\end{equation*}%
Simplifying the above expression further results in%
\begin{eqnarray*}
P &=&Q=cY^{2}+\left( c^{2}+2bc-b^{2}\right) Y-b^{2}\left( 2b+c\right)  \\
&=&\left( cY-b^{2}\right) \left( Y+2b+c\right) =z\left( Y+2b+c\right) 
\end{eqnarray*}%
Trigonometric formulas give%
\begin{equation*}
\sin \left( m+2\right) \theta -\sin \left( m+1\right) \theta =2\sin \left( 
\tfrac{\theta }{2}\right) \cos \left( 2m+3\right) \tfrac{\theta }{2}.
\end{equation*}%
Finally, we obtain%
\begin{equation*}
\left. 
\begin{array}{l}
\Delta _{m+1}^{\left( 1\right) }\left( \alpha ,\beta ,bd,bd\right) = \\ 
2\left( Y+2b+c\right) z^{m+1}\tfrac{\sin \left( \frac{\theta }{2}\right)
\cos \left( 2m+3\right) \frac{\theta }{2}}{\sin \theta },%
\end{array}%
\right. 
\end{equation*}%
which can be written%
\begin{equation*}
\left. \Delta _{m}^{\left( 1\right) }\left( \alpha ,\beta ,bd,bd\right)
=\left( Y+2b+c\right) z^{m}\tfrac{\cos \left( 2m+1\right) \frac{\theta }{2}}{%
\cos \frac{\theta }{2}}.\right. 
\end{equation*}%
This gives (\ref{Eig. 1}).\newline
For $n$ even, we obtain%
\begin{equation*}
\left. 
\begin{array}{l}
\Delta _{m+2}^{\left( 2\right) }\left( \alpha ,\beta ,bd,bd\right) =\Delta
_{m+2}^{\left( 2\right) }\left( bd,bd\right) -\alpha \Delta _{m+1}^{\left(
1\right) }\left( bb,bd\right)  \\ 
-\beta \Delta _{m+1}^{\left( 1\right) }\left( bd,bb\right) +\alpha \beta
\Delta _{m+1}^{\left( 2\right) }\left( bb,bb\right) ,\newline
\end{array}%
\right. 
\end{equation*}%
\newline
where $\Delta _{m+2}^{\left( 2\right) }\left( bd,bd\right) ,\ \Delta
_{m+1}^{\left( 1\right) }\left( bb,bd\right) =:\Delta _{m+1}^{\left(
1\right) }\left( bd,bb\right) $ and $\Delta _{m+1}^{\left( 2\right) }\left(
bb,bb\right) $ are given by (\ref{Exp 2 bddb}), (\ref{Exp 1}) and (\ref{For
bb 2}) respectively. Replacing in the above formula, yields%
\begin{eqnarray*}
&&\left. 
\begin{array}{l}
\tfrac{b\sin \theta \Delta _{m+2}^{\left( 2\right) }\left( \alpha ,\beta
,bd,bd\right) }{z^{m}}=bz^{2}\sin \left( m+3\right) \theta +C_{m+2}\sin
\left( m+2\right)  \\ 
+\left( C_{m+1}^{\left( 1\right) }+C_{m+1}^{\left( 2\right)
}+C_{m+1}^{\left( 3\right) }\right) \sin \left( m+1\right) \theta +C_{m}\sin
m\theta ,%
\end{array}%
\right.  \\
&&
\end{eqnarray*}%
where%
\begin{equation*}
\left\{ 
\begin{array}{l}
C_{m+2}=z\left\{ 
\begin{array}{c}
b\left( 3b^{2}-2bd-c^{2}\right)  \\ 
-\left( \alpha +\beta \right) \left[ b\left( Y+c\right) -cd\right] +\alpha
\beta b%
\end{array}%
\right\} , \\ 
C_{m+1}^{1}=-bz\left( d-b\right) \left( d-3b\right) +\left( d-b\right)
^{2}bc\left( Y-c\right) , \\ 
C_{m+1}^{2}=-\left( \alpha +\beta \right) z\left[ \left( d-b\right) \left(
Y+c\right) -bc\right] , \\ 
C_{m+1}^{3}=+\left( \alpha +\beta \right) \left( d\mathbf{-}b\right) \mathbf{%
c}^{2}\left( Y-c\right) +\alpha \beta b\left( b^{2}-c^{2}\right) , \\ 
C_{m}=z\left( d\mathbf{-}b\right) \left\{ b\left( d\mathbf{-}b\right)
+c\left( \alpha +\beta \right) \right\} 
\end{array}%
\right. 
\end{equation*}%
Using formula (\ref{&}), we get%
\begin{equation*}
\resizebox{1.02 \hsize} {!} {$C_{m+1}^{\left( 2\right) }=-\left( \alpha
+\beta \right) \left\{ \begin{array}{l} \left( d\mathbf{-}b\right) \left[
\begin{array}{c} 2cz\cos \theta \\ +\left( c^{2}-b^{2}\right) Y+c\left(
b^{2}-c^{2}\right) \end{array}\right] \\ -bcz\end{array}\right\}$}
\end{equation*}%
Applying the trigonometric identity (\ref{Trig}), the constants become%
\begin{equation*}
\left\{ 
\begin{array}{l}
C_{m+2}=bz\left\{ \left( 3b^{2}-2bd-c^{2}\right) -\left( \alpha +\beta
\right) Y+\alpha \beta \right\} , \\ 
C_{m+1}^{\left( 1\right) }=-bz\left( d-b\right) \left( d-3b\right) +\left(
d-b\right) ^{2}bc\left( Y-c\right) , \\ 
C_{m+1}^{\left( 2\right) }+C_{m+1}^{\left( 3\right) }=b\left\{ \left( \alpha
+\beta \right) \left( d\mathbf{-}b\right) \left( Y-c\right) +\alpha \beta
\left( b^{2}-c^{2}\right) \right\} , \\ 
C_{m}=zb\left( d\mathbf{-}b\right) ^{2}.%
\end{array}%
\right. 
\end{equation*}%
Applying (\ref{&}) to coefficient of $\sin \left(
m+2\right) \theta $ and using the trigonometric identity 
\begin{equation*}
2\cos \theta \sin \left( m+2\right) \theta -\sin \left( m+1\right) \theta
=\sin \left( m+3\right) \theta ,
\end{equation*}%
we get%
\begin{equation*}
\left\{ 
\begin{array}{l}
C_{m+3}=bz^{2}-\left( \alpha +\beta \right) bcz \\ 
C_{m+2}=\left\{ 
\begin{array}{c}
\left( 3b^{2}-2bd-c^{2}+\alpha \beta \right) z \\ 
-\left( \alpha +\beta \right) c\left( 2b^{2}-c^{2}\right) +\left( \alpha
+\beta \right) b^{2}Y%
\end{array}%
\right\} , \\ 
C_{m+1}^{\left( 1\right) }=-bz\left( d-b\right) \left( d-3b\right) +\left(
d-b\right) ^{2}bc\left( Y-c\right) , \\ 
C_{m+1}^{\left( 2\right) }+C_{m+1}^{\left( 3\right) }=\left( \alpha +\beta
\right) \left( d-b\right) b^{2}\left( Y-c\right)  \\ 
+\left( \alpha +\beta \right) bcz+\alpha \beta b\left( b^{2}-c^{2}\right) ,
\\ 
C_{m}=bz\left( d-b\right) ^{2}.%
\end{array}%
\right. 
\end{equation*}%
We eliminate the term $C_{m}$ by applying the trigonometric identity (\ref%
{Trig}), to get%
\begin{equation*}
\left\{ 
\begin{array}{l}
C_{m+3}=bz^{2}-\left( \alpha +\beta \right) bcz \\ 
C_{m+2}=b\left\{ 
\begin{array}{c}
\left( 3b^{2}-2bd-c^{2}-\left( d\mathbf{-}b\right) ^{2}+\alpha \beta \right)
z \\ 
-\left( \alpha +\beta \right) c\left( 2b^{2}-c^{2}\right) +\left( \alpha
+\beta \right) b^{3}Y%
\end{array}%
\right\} , \\ 
C_{m+1}^{\left( 1\right) }=-bz\left[ \left( d\mathbf{-}b\right) \left( d%
\mathbf{-}3b-2\left( d\mathbf{-}b\right) \cos \theta \right) \right]  \\ 
+\left( d-b\right) ^{2}bc\left( Y-c\right) , \\ 
C_{m+1}^{\left( 2\right) }+C_{m+1}^{\left( 3\right) }=\left( \alpha +\beta
\right) \left( d\mathbf{-}b\right) b^{2}\left( Y-c\right)  \\ 
+\left( \alpha +\beta \right) bcz+\alpha \beta b\left( b^{2}-c^{2}\right) ,
\\ 
C_{m}=0.%
\end{array}%
\right. 
\end{equation*}%
Applying (\ref{&}) to the term $C_{m+1}^{\left( 1\right) }$ results%
\begin{equation*}
\left. 
\begin{array}{l}
C_{m+1}^{\left( 1\right) }=-bz\left( d\mathbf{-}b\right) \left( d\mathbf{-}%
3b\right) +\left( d\mathbf{-}b\right) ^{2}bY^{2} \\ 
-\left( d\mathbf{-}b\right) ^{2}b\left( 2b^{2}-c^{2}\right) +\left(
d-b\right) ^{2}bc\left( Y-c\right) .%
\end{array}%
\right. 
\end{equation*}%
Now the question is: under what conditions the coefficients of $\sin \left(
m+3\right) \theta $ and $\sin \left( m+1\right) \theta $ are equal.\newline
Indeed, we obtain this, if we suppose that%
\begin{equation*}
d\mathbf{-}b=c\text{ and }\alpha =\beta =-b,
\end{equation*}%
which coincides with our matrix%
\begin{equation*}
e=\left( w-1\right) ^{2},\ d=w,\ b=-w(w-1),\ \ c=w^{2}.
\end{equation*}%
In this case, we have%
\begin{equation*}
\resizebox{1.02 \hsize} {!} {$\tfrac{\sin \theta \Delta _{m+2}^{\left(
2\right) }\left( \alpha ,\beta ,bd,bd\right) }{z^{m}}=D_{1}\sin \left(
m+3\right) \theta +D_{2}\sin \left( m+2\right) \theta +D_{1}\sin \left(
m+1\right) \theta ,$}
\end{equation*}%
where%
\begin{equation*}
\left\{ 
\begin{array}{l}
D_{1}=z^{2}+2bcz, \\ 
D_{2}=\left( 2b^{2}-2bc-2c^{2}\right) z+2bc\left( 2b^{2}-c^{2}\right)
-2b^{3}Y.%
\end{array}%
\right. 
\end{equation*}%
The application of the trigonometric identity%
\begin{equation*}
\sin (m+1)\theta +\sin \left( m+3\right) \theta =2\cos \theta \sin \left(
m+2\right) \theta ,
\end{equation*}%
and the formula (\ref{&}) together, give the more simplified for the
characteristic polynomial of the matrices $A_{2m}\left( \alpha ,\beta
,e,bd,c,bd\right) $%
\begin{equation*}
\Delta _{m}^{\left( 2\right) }\left( \alpha ,\beta ,bd,bd\right)
=z^{m-1}\left( \left( Y+b\right) ^{2}-\left( b+c\right) ^{2}\right) \frac{%
\sin m\theta }{\sin \theta },
\end{equation*}%
This gives (\ref{Eig. 2}) and ends the proof of the Theorem.
\end{proof}
\section{Explicit formulas for Convergence Rate}
In this section, we compute the explicit expressions of convergence rate for one-dimensional lattice networks. We study the one-dimensional lattice networks for $n=odd$ and $n=even$ cases.
\subsection{For $n$=odd}
Comparing the expressions of (\ref{2}) and (\ref{1}), we observe $%
c=w^{2},\ \ d=w,\ \ b=-w^{2}+w,\ \ \ e-\alpha =e-\beta =1-w\rightarrow
e=\left( w-1\right) ^{2},\ $\newline
$\alpha =\beta =\left( w-1\right) ^{2}-1+w=w^{2}-w=-b$, $d-b=w^{2}=c,$ Since 
$e=\left( w-1\right) ^{2}$ and $Y=e-\lambda $, then%
\begin{equation*}
\lambda _{1,0}=1.
\end{equation*}%
The other eigenvalues are the roots of the algebraic equation%
\begin{equation*}
\lambda ^{2}-2\left[ \left( w-1\right) ^{2}-w^{2}\cos \tfrac{\left(
2k+1\right) \pi }{2n}\right] \lambda +\left( 2w-1\right) ^{2}=0,
\end{equation*}%
$\ \ k=0,1,2,...,\frac{n-3}{2},$ which can be written%
\begin{equation*}
\lambda ^{2}-2\left[ -2w+1+2w^{2}\sin ^{2}\tfrac{\left( 2k+1\right) \pi }{2n}%
\right] \lambda +\left( 2w-1\right) ^{2}=0,
\end{equation*}%
$\ \ k=0,1,2,...,\frac{n-3}{2}.$ The characteristic of the above equation is%
\begin{equation*}
\Delta ^{\prime }=4w^{2}\sin ^{2}\tfrac{\left( 2k+1\right) \pi }{2n}\left[
w^{2}\sin ^{2}\tfrac{\left( 2k+1\right) \pi }{2n}-2w+1\right] .
\end{equation*}%
This gives the expressions of the eigenvalues%
\begin{equation*}
\left. 
\begin{array}{l}
\lambda _{k}=-2w+1+2w^{2}\sin ^{2}\tfrac{\left( 2k+1\right) \pi }{2n} \\ 
\pm 2w\left\vert \sin \tfrac{\left( 2k+1\right) \pi }{2n}\right\vert \sqrt{%
w^{2}\sin ^{2}\tfrac{\left( 2k+1\right) \pi }{2n}-2w+1}.%
\end{array}%
\right. 
\end{equation*}%
The largest eigenvalue is obtained when $\sin ^{2}\frac{\left( 2k+1\right)
\pi }{2n}=1$ which means $k=\frac{n-1}{2}$, i.e.$\ \lambda _{1,0}=1.$
Consequently, the second largest eigenvalue is obtained for $k=\frac{n-3}{2}.
$ That is 
\begin{equation*}
\left. 
\begin{array}{l}
\lambda _{\frac{n-3}{2}}^{+}=-2w+1+2w^{2}\sin ^{2}\tfrac{\left( n-2\right)
\pi }{2n} \\ 
+2w\left\vert \sin \tfrac{\left( n-2\right) \pi }{2n}\right\vert \sqrt{%
w^{2}\sin ^{2}\tfrac{\left( n-2\right) \pi }{2n}-2w+1}.%
\end{array}%
\right. 
\end{equation*}%
Therefore, convergence rate for $n$=odd is expressed as
\begin{equation*}
\left. 
\begin{array}{l}
\lambda _{\frac{n-3}{2}}^{+}=-2w+2w^{2}\sin ^{2}\tfrac{\left( n-2\right)
\pi }{2n} \\ 
+2w\left\vert \sin \tfrac{\left( n-2\right) \pi }{2n}\right\vert \sqrt{%
w^{2}\sin ^{2}\tfrac{\left( n-2\right) \pi }{2n}-2w+1}.%
\end{array}%
\right. 
\end{equation*}%
Convergence rate of average periodic gossip algorithm ($w=\frac{1}{2}$) is expressed as
$R=1-\sin ^{2}\tfrac{\left(
n-2\right) \pi }{2n}$.
\subsection{For $n$=even}
By comparing the (\ref{3}) and (\ref{1.1}), we observe\newline
$c=w^{2}$, $\alpha =\beta =w^{2}-w$, $d=w$, $e=(w-1)^{2}$, $b=-w^{2}+w$.
When $b=-w^{2}+w$ and $c=w^{2}$, we get 
\begin{equation*}
Y_{1,0}=w^{2},\ \ Y_{2,0}=w^{2}-2w,
\end{equation*}%
and%
\begin{equation*}
Y_{i,k}^{2}-2w^{2}Y_{i,k}\cos \tfrac{2k\pi }{n}-2\left( -w^{2}+w\right)
^{2}\left( 1-\cos \tfrac{2k\pi }{n}\right) +w^{4}=0,
\end{equation*}%
$\ \ k=1,2,...,m-1.$ Since $e=\left( w-1\right) ^{2}$ and $Y=e-\lambda $,
then%
\begin{equation*}
\left. 
\begin{array}{c}
\lambda _{1,0}=\left( w-1\right) ^{2}-w^{2}=-2w+1,\ \  \\ 
\lambda _{2,0}=\left( w-1\right) ^{2}-\left( w^{2}-2w\right) =1,%
\end{array}%
\right. 
\end{equation*}%
and the other eigenvalues are the roots of the algebraic equation%
\begin{equation*}
\lambda ^{2}-2\left[ \left( w-1\right) ^{2}-w^{2}\cos \tfrac{2k\pi }{n}%
\right] \lambda +\left( 2w-1\right) ^{2}=0,
\end{equation*}%
$\ \ k=1,2,...,\frac{n-2}{2},$ which can be written%
\begin{equation*}
\lambda ^{2}-2\left[ -2w+1+2w^{2}\sin ^{2}\tfrac{k\pi }{n}\right] \lambda
+\left( 2w-1\right) ^{2}=0.
\end{equation*}%
The characteristic of the above equation is%
\begin{equation*}
\Delta ^{\prime }=4w^{2}\sin ^{2}\tfrac{k\pi }{n}\left[ w^{2}\sin ^{2}\tfrac{%
k\pi }{n}-2w+1\right] .
\end{equation*}%
This gives the expressions of the eigenvalues%
\begin{equation*}
\left. 
\begin{array}{l}
\lambda _{k}=-2w+1+2w^{2}\sin ^{2}\tfrac{k\pi }{n} \\ 
\pm 2w\left\vert \sin \tfrac{k\pi }{n}\right\vert \sqrt{w^{2}\sin ^{2}\tfrac{%
k\pi }{n}-2w+1},%
\end{array}%
\right. 
\end{equation*}%
$k=1,2,...,\frac{n-2}{2}$. The largest eigenvalue is obtained when $\sin ^{2}%
\frac{k\pi }{n}=1$ which means $k=\frac{n}{2}$, i.e. $\ \lambda _{2,0}=1$.
Consequently, we obtain the second largest eigenvalue for $k=\frac{n-2}{2}.$
That is 
\begin{equation*}
\left. 
\begin{array}{l}
\lambda _{\frac{n-2}{2}}^{+}=-2w+1+2w^{2}\sin ^{2}\tfrac{\left( n-2\right)
\pi }{2n} \\ 
+2w\left\vert \sin \tfrac{\left( n-2\right) \pi }{2n}\right\vert \sqrt{%
w^{2}\sin ^{2}\tfrac{\left( n-2\right) \pi }{2n}-2w+1}.%
\end{array}%
\right. 
\end{equation*}%
Therefore, convergence rate for $n$=even is expressed as
\begin{equation*}
\left. 
\begin{array}{l}
R=-2w+2w^{2}\sin ^{2}\tfrac{\left( n-2\right)
\pi }{2n} \\ 
+2w\left\vert \sin \tfrac{\left( n-2\right) \pi }{2n}\right\vert \sqrt{%
w^{2}\sin ^{2}\tfrac{\left( n-2\right) \pi }{2n}-2w+1}.%
\end{array}%
\right. 
\end{equation*}%
Convergence rate of average periodic gossip algorithms ($w=\frac{1}{2}$)  is expressed as $\lambda _{\frac{n-2}{2}}^{+}=1-\sin ^{2}\tfrac{\left(
n-2\right) \pi }{2n}$.\\
\section{Effect of Link Failures on Convergence Rate}
\subsection{For $n$=even}
Comparing the expressions of (\ref{4}) and (\ref{1}), we observe 
$c=\frac{(p-1)^2}{4}$,\ \ $d=\frac{1-p}{2}$,\ \ $b=\frac{1-p^2}{4}$,\ \ \ $e=\frac{(p+1)^2}{4}$, \ \ $\alpha=\beta=\frac{p^2-1}{4}$ and $Y=e-\lambda$, 
then
\begin{equation*}
\lambda _{1,0}=1.
\end{equation*}
The other eigenvalues are the roots of the algebraic equation
\begin{equation}
\resizebox{1.02 \hsize} {!} {$\lambda ^2  + \lambda \left( { - 2e + 2c\cos \frac{{(2k + 1)\pi }}{{2m + 1}}} \right) + e^2  - 2ce\cos \frac{{(2k + 1)\pi }}{{2m + 1}} + 2b^2 \cos \frac{{(2k + 1)\pi }}{{2m + 1}} - 2b^2  + c^2  = 0$}
\end{equation}
Substituting the values of $c$, $b$, and $e$ results in
\begin{equation}
\resizebox{1.02 \hsize} {!} {$\lambda ^2  + \lambda \left( { - \frac{{(p + 1)^2 }}{2} + \frac{{(p - 1)^2 }}{2} - (p - 1)^2 \sin ^2 \frac{{(2k + 1)\pi }}{{4m + 2}}} \right) + p^2  = 0$}
\end{equation}
Then, $\lambda_k$ is expressed as
\begin{equation}
\resizebox{1.02 \hsize} {!} {$\lambda _k  = p + \frac{{(p - 1)^2 }}{2}\sin ^2 \frac{{(2k + 1)\pi }}{{4m + 2}} + \sqrt {\frac{{(p - 1)^4 }}{4}\sin ^4 \frac{{(2k + 1)\pi }}{{4m + 2}} + p(p - 1)^2 \sin ^2 \frac{{(2k + 1)\pi }}{{4m + 2}}}$}
\end{equation}
Second largest eigenvalue is expressed as
\begin{equation}
\resizebox{1.02 \hsize} {!} {$\lambda _{\frac{n-2}{2}}^{+}  = p + \frac{{(p - 1)^2 }}{2}\sin ^2 \frac{{(n - 2)\pi }}{{2n}} + \sqrt {\frac{{(p - 1)^4 }}{4}\sin ^4 \frac{{(n - 2)\pi }}{{2n}} + p(p - 1)^2 \sin ^2 \frac{{(n - 2)\pi }}{{2n}}}$}
\end{equation}
Therefore, convergence rate of the average periodic gossip algorithms when communication links fail with the probability $p$ for even number of nodes is expressed as
\begin{equation}
\resizebox{1.02 \hsize} {!} {$R  = 1-p - \frac{{(p - 1)^2 }}{2}\sin ^2 \frac{{(n - 2)\pi }}{{2n}} - \sqrt {\frac{{(p - 1)^4 }}{4}\sin ^4 \frac{{(n - 2)\pi }}{{2n}} - p(p - 1)^2 \sin ^2 \frac{{(n - 2)\pi }}{{2n}}}$}
\end{equation}
\subsection{For $n$=odd}
Comparing the expressions of (\ref{5}) and (\ref{1.1}), we observe 
$c=\frac{(p-1)^2}{4}$,\ \ $d=\frac{1-p}{2}$,\ \ $b=\frac{1-p^2}{4}$,\ \ \ $e=\frac{(p+1)^2}{4}$, \ \ $\alpha=\beta=\frac{p^2-1}{4}$ and $Y=e-\lambda$, 
then
\begin{equation*}
\lambda _{1,0}=1.
\end{equation*}
The other eigenvalues are the roots of the algebraic equation
\begin{equation}
\resizebox{1.02 \hsize} {!} {$\lambda ^2  + \lambda \left( { - 2e + 2c\cos \frac{{k\pi }}{m}} \right) + e^2  - 2ce\cos \frac{{k\pi }}{m} + 2b^2 \cos \frac{{k\pi }}{m} - 2b^2 \cos \frac{{k\pi }}{m} - 2b^2  + c^2  = 0$}
\end{equation}
Substituting the values of $c$, $b$, and $e$ results in
\begin{equation}
\resizebox{1.02 \hsize} {!} {$\lambda ^2  + \lambda \left( { - \frac{{(p + 1)^2 }}{2} + \frac{{(p - 1)^2 }}{2} - (p - 1)^2 \sin ^2 \frac{{k\pi }}{{n}}} \right) + p^2  = 0$}
\end{equation}
Then, $\lambda_k$ is expressed as
\begin{equation}
\resizebox{1.02 \hsize} {!} {$\lambda _k  = p + \frac{{(p - 1)^2 }}{2}\sin ^2 \frac{{k\pi }}{n} + \sqrt {\frac{{(p - 1)^4 }}{4}\sin ^4 \frac{{k\pi }}{{n}} + p(p - 1)^2 \sin ^2 \frac{{k\pi }}{n}}$}
\end{equation}
Second largest eigenvalue is expressed as
\begin{equation}
\resizebox{1.02 \hsize} {!} {$\lambda _{\frac{n-3}{2}}^{+}  = p + \frac{{(p - 1)^2 }}{2}\sin ^2 \frac{{(n - 2)\pi }}{{2n}} + \sqrt {\frac{{(p - 1)^4 }}{4}\sin ^4 \frac{{(n - 2)\pi }}{{2n}} + p(p - 1)^2 \sin ^2 \frac{{(n - 2)\pi }}{{2n}}}$}
\end{equation}
Therefore, convergence rate of the average periodic gossip algorithms when communication links fail with the probability $p$ for odd number of nodes is expressed as
\begin{equation}
\resizebox{1.02 \hsize} {!} {$R  = 1-p - \frac{{(p - 1)^2 }}{2}\sin ^2 \frac{{(n - 2)\pi }}{{2n}} - \sqrt {\frac{{(p - 1)^4 }}{4}\sin ^4 \frac{{(n - 2)\pi }}{{2n}} - p(p - 1)^2 \sin ^2 \frac{{(n - 2)\pi }}{{2n}}}$}
\end{equation}

\begin{table}[]
\centering
\caption{Convergence Rate of Small Scale WSNs}
\label{my-label}
\begin{tabular}{|l|l|l|ll}
\cline{1-3}
\textbf{Number of Nodes} & \textbf{Convergence Rate} & \textbf{Optimal Gossip Weight} &  &  \\ \cline{1-3}
4                        & 0.8                        & 0.6                 &  &  \\ \cline{1-3}
5                        & 0.6                      & 0.7                &  &  \\ \cline{1-3}
6                        & 0.6                       & 0.7                 &  &  \\ \cline{1-3}
7                       & 0.6                      & 0.7                &  &  \\ \cline{1-3}
8                        & 0.4                     & 0.8                  &  &  \\ \cline{1-3}
9                        & 0.4                       & 0.8                 &  &  \\ \cline{1-3}
10                        & 0.4                        & 0.8                  &  &  \\ \cline{1-3}
11                        & 0.4                      & 0.8                &  &  \\ \cline{1-3}
12                       & 0.4                      & 0.8                 &  &  \\ \cline{1-3}
13                        & 0.3034                        & 0.8                  &  &  \\ \cline{1-3}
14                        & 0.2412                    & 0.8              &  &  \\ \cline{1-3}
15                       & 0.2015                   & 0.8                &  &  \\ \cline{1-3}
16                       & 0.2                    & 0.9                &  &  \\ \cline{1-3}
17                       & 0.2                     & 0.9                 &  &  \\ \cline{1-3}
18                       & 0.2                       & 0.9                 &  &  \\ \cline{1-3}
19                       & 0.2                       & 0.9                  &  &  \\ \cline{1-3}
20                       & 0.2                      & 0.9                 &  &  \\ \cline{1-3}
\end{tabular}
\end{table}
\begin{table}[]
\centering
\caption{Convergence Rate of Large Scale WSNs}
\label{my-label}
\begin{tabular}{|l|l|l|ll}
\cline{1-3}
\textbf{Number of Nodes} & \textbf{Convergence Rate} & \textbf{Optimal Gossip Weight} &  &  \\ \cline{1-3}
100                        & 0.009                       & 0.9                 &  &  \\ \cline{1-3}
200                      & 0.0022                      & 0.9                &  &  \\ \cline{1-3}
300                      & 0.001                     & 0.9                 &  &  \\ \cline{1-3}
400                     & 0.0006                    & 0.9              &  &  \\ \cline{1-3}
500                      & 0.1                   & 0.9                  &  &  \\ \cline{1-3}
600                      & 0.002                       & 0.9                 &  &  \\ \cline{1-3}
700                       & 0.002                      & 0.9                 &  &  \\ \cline{1-3}
800                       & 0.001                     & 0.9              &  &  \\ \cline{1-3}
900                      & 0.001                   & 0.9                &  &  \\ \cline{1-3}
1000                       & 0.0001                       & 0.9                  &  &  \\ \cline{1-3}
\end{tabular}
\end{table}

\section{Numerical Results}

In this section, we present the numerical results. Fig. 2 shows the convergence rate versus number of nodes in one-dimensional lattice networks for average periodic gossip algorithms(\textit{w} =0.5). We have observed that convergence rate reduces exponentially with the increase in number of nodes. In every time step, nodes share information with their direct neighbors  to achieve the global average. Thus, for larger number of nodes, more time steps will be required, thereby leading to slower convergence rates. As shown in Table. I, optimal gossip weights are varying with the number of nodes until \textit{n}=16 and it's value becomes 0.9 from $\textit{n}\ge 16$. Fig. 3 shows the convergence rate versus gossip weights for large-scale networks. We have observed that for large-scale lattice networks, the optimal gossip weight turns out to be 0.9 (see Table. 2). Hence, we can conclude that for any reasonable sized network ($\textit{n}\ge 16$), gossip weight should be 0.9 for achieving faster convergence rates in one dimensional lattice networks. We measure the efficiency of periodic gossip algorithms at $\textit{w}$=0.9 over average periodic gossip algorithms($\textit{w}$=0.5) by using relative error $\textit{RE} = \frac{{R_{0.9}  - R_{0.5} }}{{R_{0.9} }}$. $R_{0.9}$ and $R_{0.5}$ denote the convergence rate for \textit{w}=0.9 and \textit{w}=0.5 respectively. Fig. 5 shows the relative error versus the number of nodes for small-scale networks. Here, we observed  that relative error is increasing with the number of nodes. Fig. 6 shows the relative error versus the number of nodes for large-scale networks. In this case, we observed that relative error is approximately constant (0.89), with any further increase in number of nodes. To study the effect of communication link failures, we plot the Fig. 7. We have observed that, convergence rate is decreasing with the probability of link failures.

\begin{figure}[tbp]
\centering
\includegraphics[width=6.5cm,height=5cm]{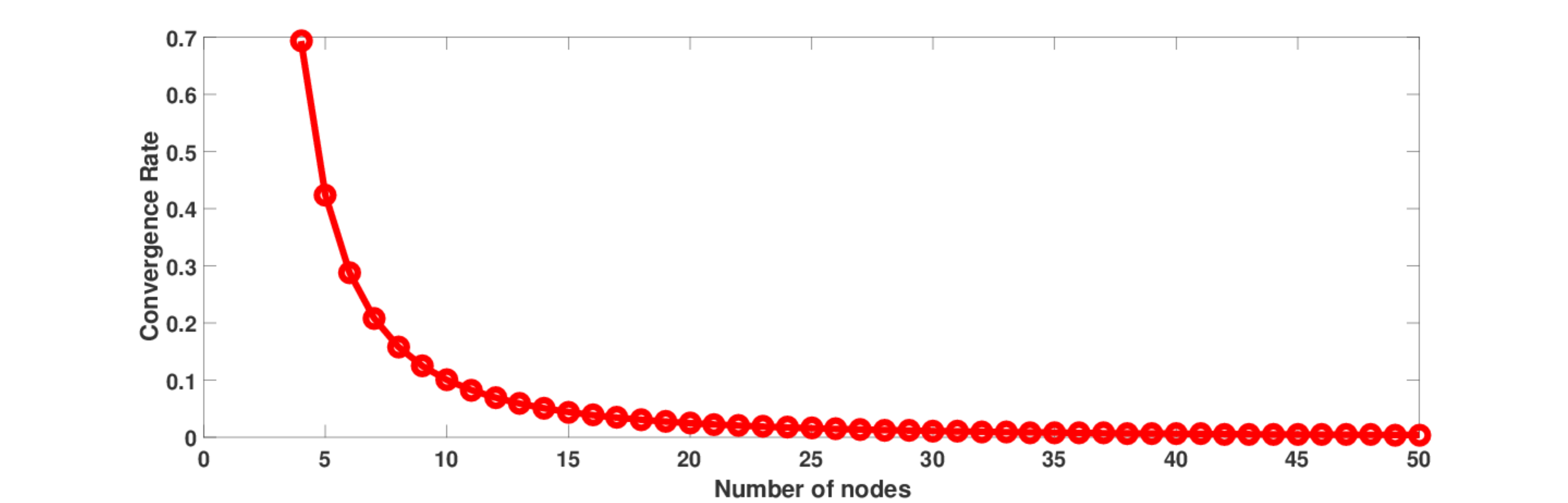}
\caption{Convergence Rate versus Number of nodes for Average Gossip Algorithm ($w$=0.5).}
\label{fig:verticalcell}
\end{figure}

\begin{figure}[tbp]
\centering
\includegraphics[width=7cm,height=5cm]{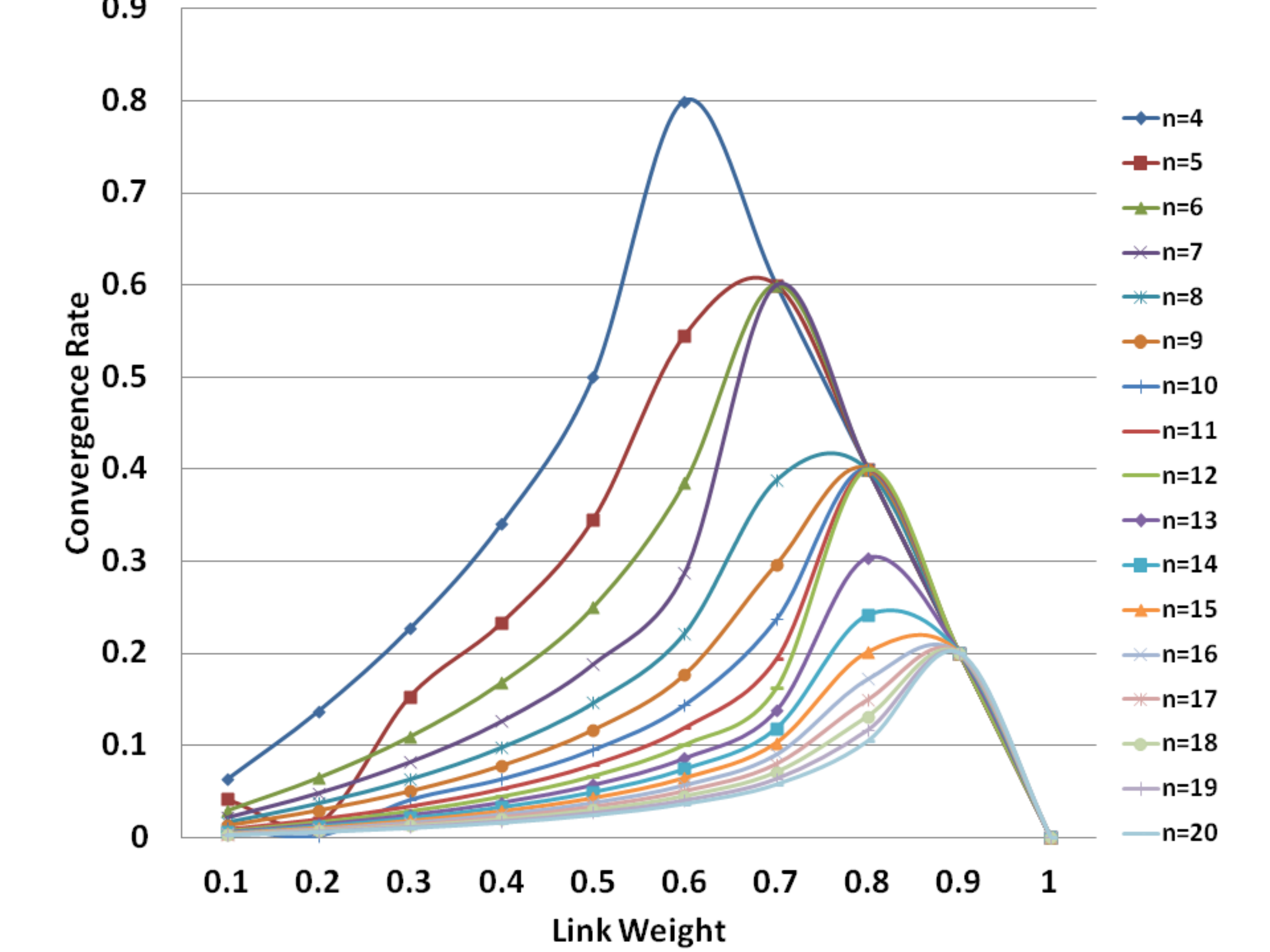}
\caption{Convergence Rate versus Gossip Weight for Small Scale WSNs.}
\label{fig:verticalcell}
\end{figure}

\begin{figure}[tbp]
\centering
\includegraphics[width=7cm,height=5cm]{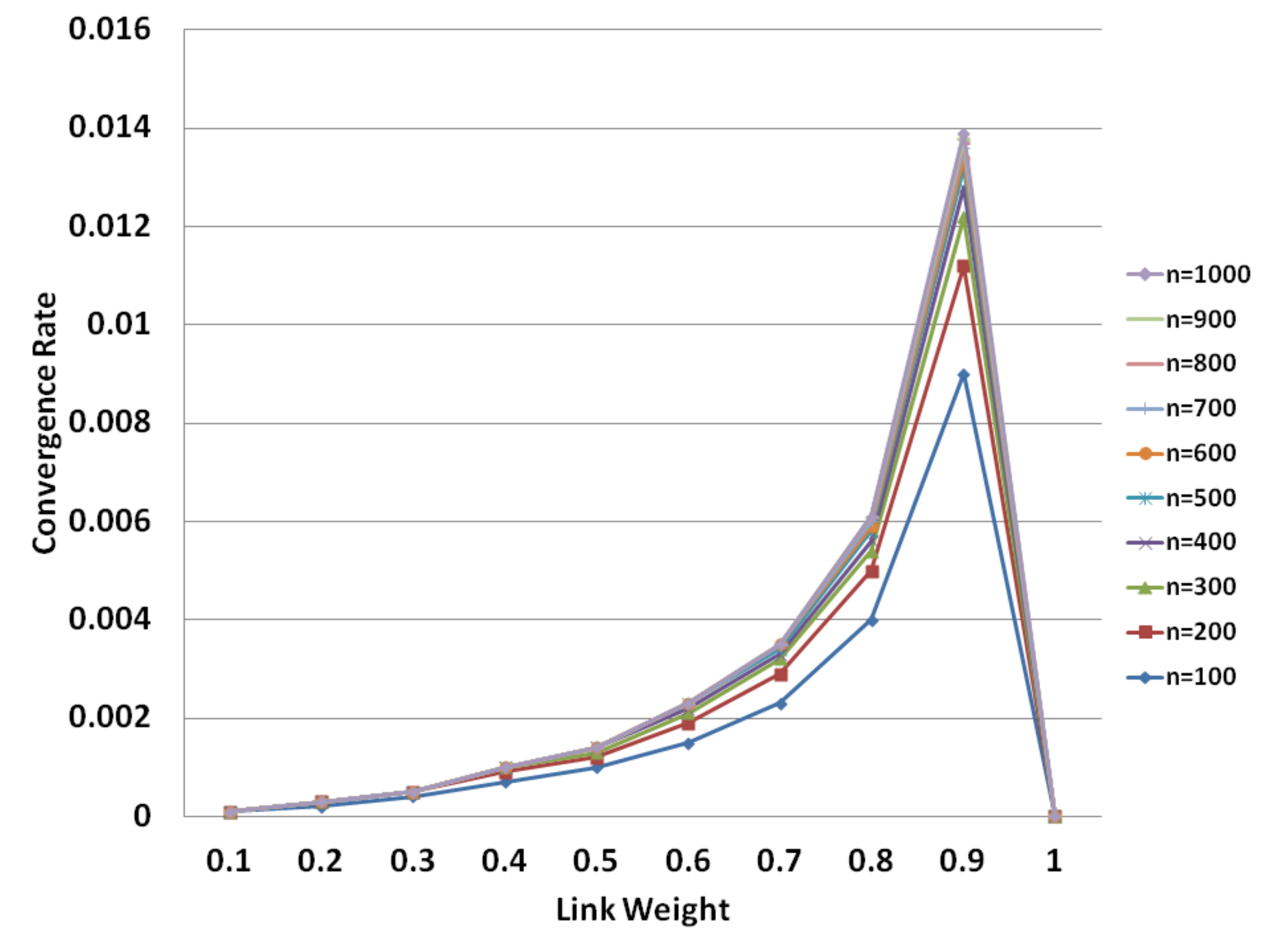}
\caption{Convergence Rate versus Gossip Weight for Large Scale WSNs.}
\label{fig:verticalcell}
\end{figure}

\begin{figure}[tbp]
\centering
\includegraphics[width=7cm,height=5cm]{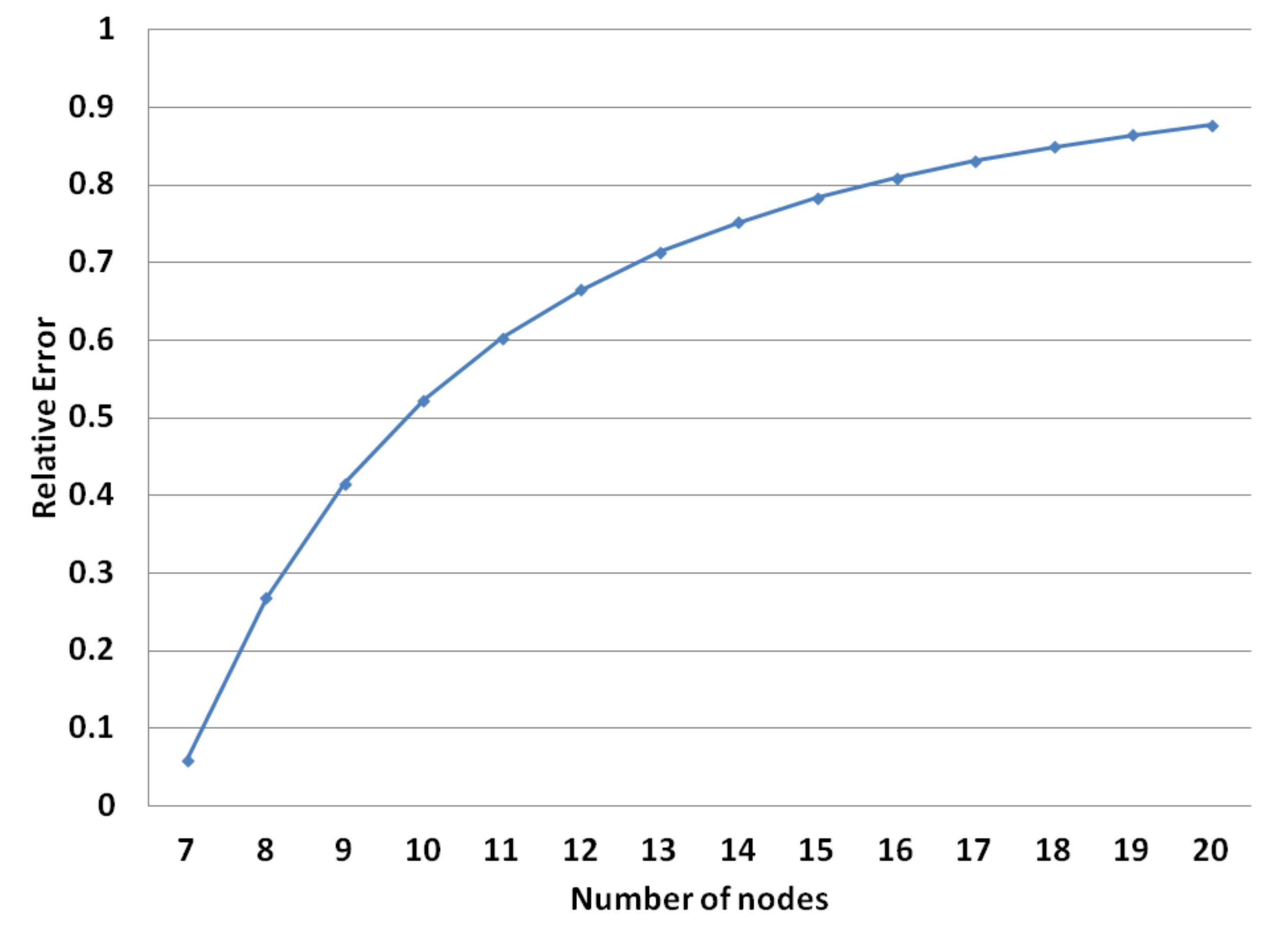}
\caption{Relative Error versus Number of Nodes for Small Scale WSNs.}
\label{fig:verticalcell}
\end{figure}

\begin{figure}[tbp]
\centering
\includegraphics[width=7cm,height=5cm]{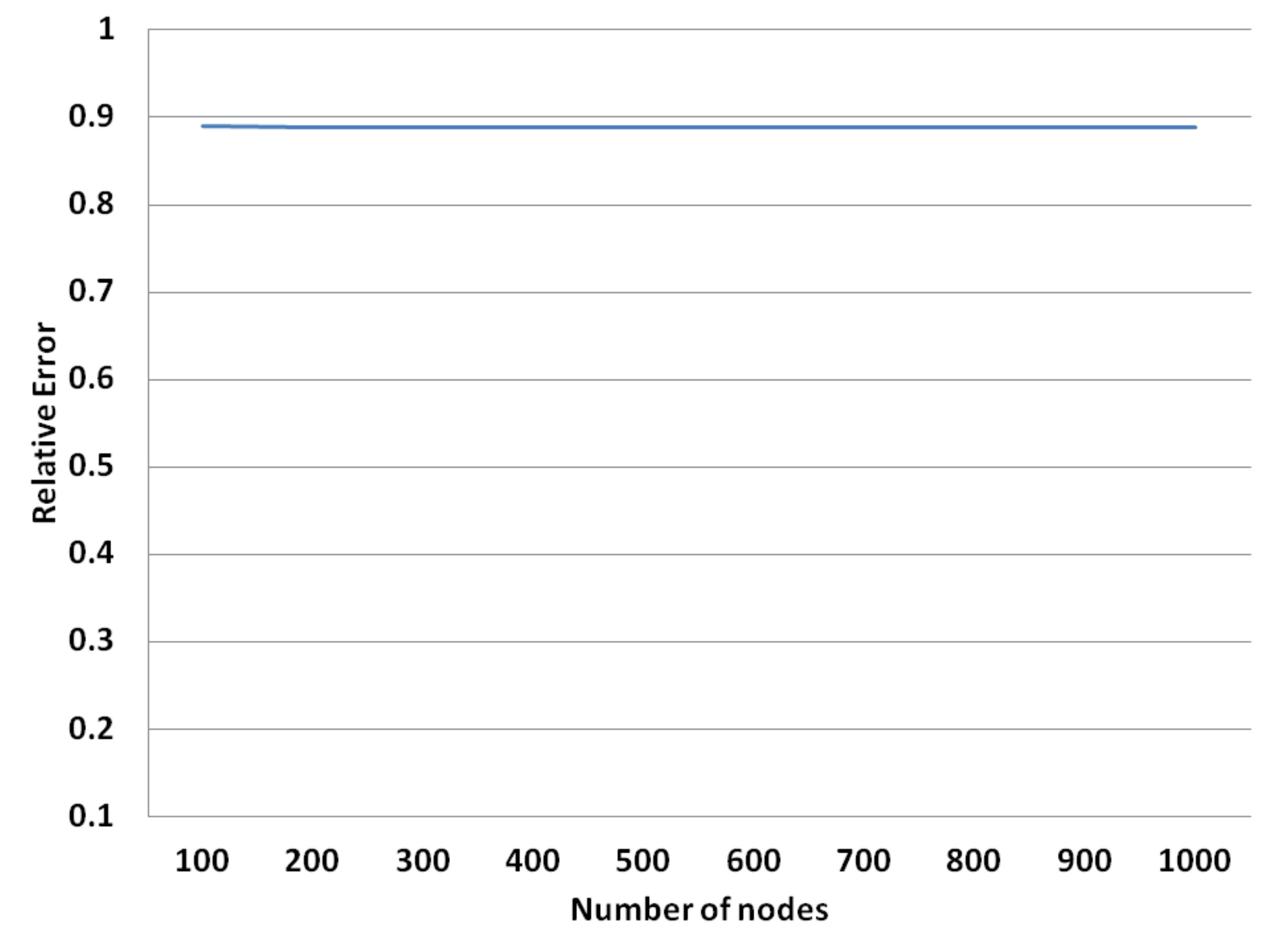}
\caption{Relative Error versus Number of Nodes for Large Scale WSNs.}
\label{fig:verticalcell}
\end{figure}

\begin{figure}[tbp]
\centering
\includegraphics[width=7cm,height=5cm]{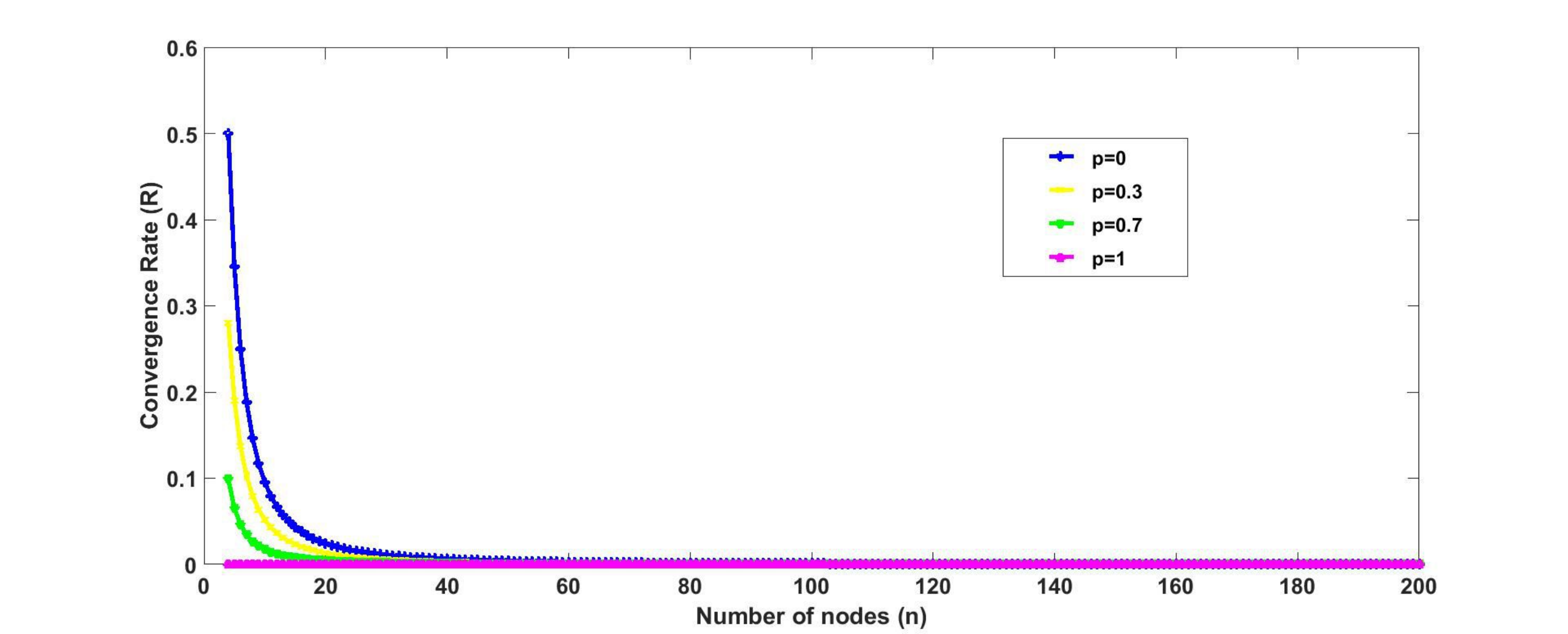}
\caption{Effect of Link Failures on Convergence Rate.}
\label{fig:verticalcell}
\end{figure}

\section{Conclusions}
Estimating the convergence rate of a periodic gossip algorithm is computationally challenging in large-scale networks. This paper derived the explicit formulas of convergence rate for one-dimensional lattice networks. Our work drastically reduces the computational complexity to estimate the convergence rate for large-scale WSNs. We also derived the explicit expressions for convergence rate in terms of gossip weight and number of nodes using linear weight updating approach. Based on our findings, we have observed that there exists an optimum gossip weight which significantly improves the convergence rate for periodic gossip algorithms in small-scale WSNs ($\textit{n}<16$). Our numerical results demonstrate that periodic gossip algorithms achieve faster convergence rate for large-scale networks ($\textit{n} \ge 16$) at $\textit{w}$=0.9 over average periodic gossip algorithms. In this work, we also considered the communication link failures and derived the closed-form expression of convergence rate for average periodic gossip algorithms. Furthermore, our formulation also helped to obtain the eigenvalues of perturbed pentadiagonal matrices. To the best of our knowledge, this is the first paper to derive the explicit formulas of eigenvalues for pentadiagonal matrices.

\section*{Acknowledgment}

The author S. Kouachi thanks the KSA superior education ministry for the
financial support.

\ifCLASSOPTIONcaptionsoff
\newpage \fi

\bibliographystyle{IEEEtran}
\bibliography{tvt}


\end{document}